\documentclass[journal]{IEEEtran}
\usepackage{dsfont}
\usepackage{graphicx}
\usepackage{dsfont}
\usepackage{enumerate}
\usepackage{color}
\graphicspath{{figures/}}

\usepackage[fleqn]{amsmath}
\usepackage{amsfonts,amssymb}
\usepackage{subeqnarray}
\usepackage{cases}
\usepackage{amsthm}

\usepackage{subfigure}
\usepackage{cite}
\usepackage{lipsum}
\usepackage{bm}

\usepackage{algorithm}
\usepackage{algorithmic}

\usepackage{caption}
\usepackage{lipsum}
\usepackage{multicol}
\graphicspath{{figures/}}
\usepackage{graphicx}
\usepackage{lipsum}
\usepackage{bm}
\usepackage{caption}
\graphicspath{{figures/}}
\usepackage{mathrsfs}
\usepackage{amsfonts,amssymb}
\usepackage{epstopdf}
\usepackage{amsmath}
\usepackage{bm}
\usepackage{mathrsfs}
\newtheorem{theorem}{Theorem}
\newtheorem{lemma}[theorem]{Lemma}
\newtheorem{corollary}[theorem]{Corollary}

\makeatletter
\renewcommand\normalsize{%
   \@setfontsize\normalsize\@xpt\@xiipt
   \abovedisplayskip 1\p@ \@plus2\p@ \@minus5\p@
   \abovedisplayshortskip \z@ \@plus3\p@
   \belowdisplayshortskip 1\p@ \@plus3\p@ \@minus3\p@
   \belowdisplayskip \abovedisplayskip
   \let\@listi\@listI}
\makeatother
\begin{document}
%



\title{Secure Communication through Wireless-Powered Friendly Jamming: Jointly Online Optimization  over Geography, Energy  and Time}


\author{\normalsize Pan Zhou, \emph{Member, IEEE}
\thanks{Pan Zhou is  from the School
of Electronic Information and Communications, Huazhong University of
Science and Technology, Wuhan 430074, China. E-mail: zhoupannewton@gmail.com
} \\}



\maketitle
\thispagestyle{empty}
\pagestyle{plain}

\begin{abstract}
Exploring the interference-emitting friendly
 jammers to protect the sensitive communications  in the presence of eavesdroppers has increasingly
    being investigated in literature.  	In parallel, scavenging energy from abient radio signals for energy-constrained devices, namely \emph{wireless energy harvesting} (WEH),    has also drawn significant attention. Without relying on external energy supply, the wireless-powered friendly jammer by WEH from
     legitimate wireless devices is an effective approach to prolong their lifetime and gain the flexibility in deployments.
     This paper studies the online optimization of the  placement and WEH of a set of friendly jammers in a geographic location with the
     energy-efficiency (EE) consideration. We adopt a simple ``time switching" protocol where power transfer and jammer-assisted secure communications occur in different time blocks when WEH requests are launched.  Our scheme has the
     following important advantages: 1) The proposed online jammers placement and interfering power
     allocation to attack eavesdroppers is the first distributed  and scalable solutions within any specified geographic region; 2) We model the WEH for jammers as a \emph{JAM-NET lifetime maximization problem}, where online
 scheduling algorithms with heterogeneous energy demands of each jammer  (from energy sources) are designed; 3) Under our model, the
    problem of placing a minimum number of jammers with distance-based power assignments is NP-hard, and
  near optimal PTAS approximation algorithms are provided;  4) When durations of the eavesdropping and legitimate communicating are
  available and the scenario is extended to the multi-channels setting, our results are strengthened to  see further improved
EE and reduced number of jammers. Simulations back up our theory.
\end{abstract}


%
\IEEEpeerreviewmaketitle


\section{Introduction}
\subsection{Background and Motivation}
Due to the inapplicability and high computational complexity of cryptography in many dynamic wireless environments, physical
layer security techniques \cite{PHY1,PHY2} for securing the transfer of highly sensitive information in wireless communications have attracted significant attention in the past decades.
 Systems such as mobile personal healthcare records \cite{EXMp1},
  contactless payment cards  \cite{EXMp2}, telemedicine systems using wireless networks \cite{EXMp3} and military
 sensor networks \cite{EXMp4} all employ wireless technologies to transmit potentially sensitive information.  In particular, placing  jammers as cooperative communication nodes has recently been explored as an effective means to achieve the secure wireless communications from eavesdroppers \cite{ED1, ED2}.
 By exploiting the shared nature of wireless channels, the successful deployments of  jammers for security must
achieve the twin goals that i) reducing the Signal-to-Interference-plus-Noise Ratio (SINR) of eavesdroppers to a level that far below
 a threshold for successful reception, and ii) maintaining the sufficient channel qualities such that the SINR at the legitimate receivers are
not reduced too much so as to prevent the reception of \emph{wireless information transfer} (WIT). However, this is often realized at the expense of additional
power consumption for friendly jammers.

Conventional energy harvesting methods rely on various renewable energy sources
in the environments, such as solar, wind, vibration and thermoelectric, that are usually unstable and uncontrollable. In contrast, the
recent advance in radio frequency (RF) enabled \emph{wireless power transfer} (WPT) technology provides an attractive solution by powering
wireless nodes with continuous and stable energy over the air \cite{WEHM_0}. The key idea of this technology is by
leveraging the far-field radiative properties of electromagnetic  wave (EMW), the wireless nodes could capture EMW remotely from RF signals and convert it
into direct current to charge its battery.

 Recently, the WPT has attracted great interests  in the
research community on energy constrained wireless networks. In \cite{E_TR_1,E_TR_2,E_TR_3,E_TR_4}, the authors studied the
sources simultaneously performing the WPT and WIT to destinations  and problem that how the wireless nodes makes use of the
harvested energy from WPT to enable communications. Motivated by these works, the process of WET can be
fully controlled, hence it is preferred to be applied in wireless networks with critical quality-of-service requirements, such as
 secure wireless communications.  In \cite{EJ_TR_1, EJ_TR_2, EJ_TR_3}, the authors considered secure communications with the existence
of a single information receiver and several wireless energy-harvesting eavesdroppers. In \cite{NOUSEEJ_TR_0}, the authors presented
the coexistence of three  types of destination in a simple wireless communication scenario: an information
receiver,  an eavesdropper and a  harvesting wireless energy receiver. As noticed, all these works \cite{EJ_TR_1, EJ_TR_2, EJ_TR_3, NOUSEEJ_TR_0} on
only focus on the process of energy-harvesting, the use of which at the receivers (e.g., friendly jammers) for  secure communications
is not studied. Recently, work \cite{FJ_TWC_16} used the harvested power at a friendly jammer as a useful resource
to emit constructive interference to attack the eavesdroppers that secures the legitimate wireless communication link for the first time. However, their
focus is only on a single communication link, where the placement of multiple friendly jammers in a geographical locations and and the  energy efficiency
 (EE) issue in the power management  for general wireless networks applications are not studied yet.


Similar to  \cite{Geo1, Geo2, Geo3}, we consider the following typical scenario that the legitimate communication  is often conducted in some restricted geographic locations, where jammers placed in the vicinity are used to secure the legitimate communication.     Different from  offline (centralized) solution
\cite{Geo1, Geo2, Geo3} that the transmission power and number of jammers are required to be optimized separately over a
known geographic region, we do not restrict the scale and geometry of the geographic locations and emphasize the potential dynamics of
nodes within the secure communication. Hence, distributed and online jammers placement protocols are desirable. In this case, the friendly
jammers have the flexibility  to be placed randomly at any feasible locations. Therefore, their lifetime is usually constrained to the energy stored in the battery, and WEH as a
promising approach  \cite{WEH_1,WEH_2,WEH_3} is demanding to prolong their lifetime. Thus, it is highly motivating to study the EE
 by leveraging both the WEH and  the interfering power allocation processes  in the power management.


\subsection{Our Work and Contribution}
As illustrated in Fig. 1,  the legitimate
communications within a region, named as \emph{storage} $\mathfrak{S}$, surrounded by a \emph{fence} $\mathfrak{F}$,  and the jammers
are placed between the space of $\mathfrak{S}$ and $\mathfrak{F}$ to protect the legitimate communications from eavesdroppers lying outside the fence.  To this end,
the friendly jammers act as  passive security assistants of legitimate communication links. The
deployments of this low cost and simple passive jammers brings both the important advantages and challenges: on the one hand, the
WEH-based scheme without any power line connection facilitate the flexibly online  jammer deployments. Such placements are inherently local and particularly useful in distributed deployments, which is highly desirable for complex geographic areas, e.g., lofty and rugged hills, rough grounds, pot-holed city streets and architectures, etc. and large-scale network deployments;
on the other hand,  jammers should have low design cost and complexity as well as have high efficiency in
energy harvesting method to enable its functionality.
Moreover,
 the jammers are not capable to communicate with each other and can only passively report their ``remaining energy status" (as ``energy demands" from the
perspective of legitimate networks) periodically to the transmitters.   In these settings, when a request of placing a new jammer targeting on a passive eavesdropper (or potential eavesdropping position) arrives, an online
algorithm needs to decide whether to accept the request and assign a transmission power (one out of $F$ channels in the multi-channel setting) to it. Decisions
about the acceptance as well as the power and channel assignments cannot be revoked later.

To solve the above  secure communications problem, we propose to use a set of wireless-powered friendly jammers as a
defensive and constructive interference-emitting companions, where jammers harvest energy via WPT
from the legitimate source nodes. The energy harvesting circuit of the jammers (e.g., consisting of a passive
low-pass filter and diode(s) \cite{ZhouR13}) is very simple and cost effective, and such a configuration is
very easy to be controlled by the external energy sources. We use a simple ``time switching" scheme \cite{E_TR_4,Nas15} such that
there are two phases within a complete secure communication circle: namely WPT and WIT for secure communication. In the first phase, due to locations
varies over time and  CSI is not available from the passive jammers  to the energy sources and different energy demand of
each jammer, it is challenging to find the optimal energy scheduling algorithm for WPT.
%


\begin{figure}
\centering
\includegraphics[scale=.45]{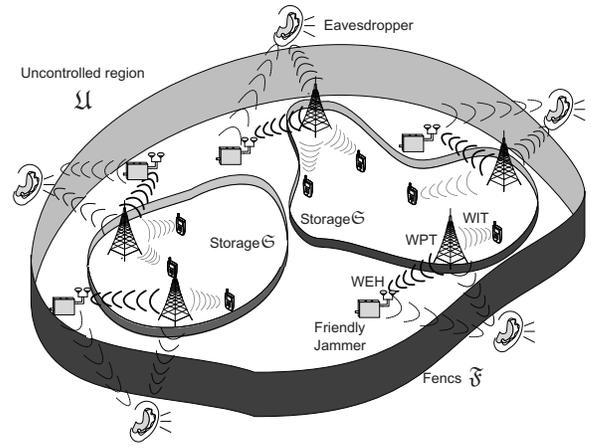}
\caption{Friendly jammers-assisted secure communications in the geographic region against
eavesdroppers outside the fence. Jammers are placed between the fence  $\mathfrak{S}$ and the storage
 $\mathfrak{F}$.  }
\label{fig:digraph}
\vspace{-.2cm}
\end{figure}

In the second phase, using
competitive analysis, we study algorithms using distance-based power assignments from each placed jammer to the
eavesdropper location(s). Accepted placement request must satisfies constraints on the SINR for both eavesdroppers and legitimate receivers. The
aim is to minimize the number of accepted placement requests of jammers.
We first focus on the case of a single
channel for the request sets with spatial lengths in $[1, \Delta]$ and average minimal duration between legitimate WIT and
eavesdropping in $[1, \Lambda]$.

The main contributions of this work are summarized below:

\emph{1)}. The novelty of the work lies in the design of the first distributed
protocol that provides secure communication in any geographically restricted communication networks using
energy-constrained friendly jammers wirelessly powered by legitimate transmitters as energy sources.

\emph{2)}. We consider the energy-efficiency (EE) during the
whole design. We adopt a time-division based protocol with during the WPT phase, where legitimate transmitters only provide quota  total energy $\bar E$ within a total WEH micro-slots of $T$. The problem is named as the \emph{JAM-NET lifetime maximization}. It is formulated by an adaptive
    constrained integer linear program (ILP) to meet the goal of EE with heterogenous energy demands.
       We analysis the originally hard problem from several aspects with practical implementation considerations
        by advanced online learning algorithms, which indicates that the optimal online scheduling algorithms are available.
         In addition, for constructive jamming power assignment, we studied the linear
power assignment based on the previous \emph{distance-based} power assignment policy, which has the advantage of being energy-minimal.

\emph{3)}. The friendly jammers only need to know minimal
information about the communication taking place. They are proactive rather than reactive, requiring no overhead for the
legitimate communication nodes and no synchronization amongst themselves. Our protocol supports dynamic behaviors, e.g.,
mobility, eavesdropping (communicating) completion or addition/removal of nodes, as along the secure communication are restricted
to the storage. However, our proposed  protocol is adaptive to the situations such information is available, e.g.,  exact positions  and frequency of both legitimate communications and
eavesdropping behaviors,
 and foreseen further EE improvements and reduced number of jammers.

\emph{4)}. We indicate that it is NP-hard to minimize the number of placed jammers necessary to protect the
geographic domain of secure communications.
 \begin{itemize}
   \item  We derive, for any fixed $\varepsilon$, the upper bound of $O((1+\varepsilon) {\Delta ^{d/2}})$ and $O((1+\varepsilon) \Lambda {\Delta ^{d/2}})$ on the
competitive ratio of any deterministic online algorithm without and with the knowledge of duration $\Lambda$.
   \item  Then, we extend the result
to the general polynomial power assignment with parameter $r$ that
cannot yield a competitive ratio worsen than $O((1+\varepsilon) {\Delta ^{\min\{r, 1-r\}}})$; for the
square root power assignment, it yields an upper bound of $O((1+\varepsilon) {\Delta ^{d/2}})$. In fact, we show that
 this bound holds for any distance-based power assignment.
   \item Our upper bounds reveals an exponential gap of the achievable approximation guarantees between
 deterministic online and offline \cite{Geo1, Geo2, Geo3} algorithms. The main difficulty of the online
scenario turns out to be that the request cannot be ordered by length due to distributed deployments. Given $r\in [0,1]$, we showed that the
square root $r=1/2$ achieves near optimal competitive radio among all distance-based power assignments and it superior to any
other polynomial power assignment.
\item We extended our analysis to the multi-channel cases. We generalize the analysis of \textup{MULTI-CHAN JAM-Distance} algorithm from 1 to $F$ channels. Using $F= F' \cdot F''$ channels is only $\Omega ((1 + \varepsilon )F \cdot {\Lambda ^{1/F'}} \cdot {\Delta ^{1/F''}})$-competitive. It
 indicates an exponential reduction in the competitive ratio, which indicates that the \emph{multi-channel diversity} could improve the security
 of legitimate communications.
 \end{itemize}


\subsection{Related Work} Most related work, e.g. \cite{E_TR_1}-\cite{NOUSEEJ_TR_0}, focused on the wiretap channel \cite{WT-Channel75} in the field of
information theory, in which a single eavesdropper tries to listen to legitimate communication between a pair of
 nodes. It is shown that perfect security is possible when the eavesdropper's channel quality is lower than a threshold. Recent works
\cite{Rw1,Rw2} also focus on the MIMO wiretap channel where the transmitter, receiver and eavesdropper may configured  with multiple antennas.
In \cite{Nas15}, the authors had used a wireless-powered relay to help
the point-to-point communication. In \cite{Rw3},
the authors studied the friendly jamming signal design
to help the secure communication based on the knowledge
of the uncontrollable energy harvesting process.
 Different from \cite{Nas15} \cite{Rw3}, authors in \cite{FJ_TWC_16} considered  the WEH at a friendly jammer
 to emitting constructive jamming power for a secure communication link, where the jamming power and rate parameters are optimized
for secure communication.  However, most of these works primarily
targeted to the theoretical significant due to the simple scenario under consideration but do not explore the geometry of the problem
sufficiently.

Vilela et al. \cite{GeoRw1} showed that without any assumptions on the locations of friendly jammers and eavesdroppers, jammers
 could co-transmitting with the legitimate transmitter and in the vicinity of a common destination. The authors formulated this setting
as a graph and use ILP to find an optimal subset of jamming nodes. In \cite{GeoRw2}, the authors study the
asymptotic behaviors for jammers and eavesdroppers at the  stochastically distributed locations. In particular, they proposed
the concept of \emph{Secure Throughput}, which is based on the probability that a message is successfully received only by legitimate receivers.
To our best knowledge, \cite{Geo1, Geo2, Geo3} are the only works that adapt  friendly jammers into complex geometric positioning constraints.
They provided offline optimal solutions to jammer placement problem that involves both
continuous aspects \cite{Geo2, Geo3} (i.e., power allocations) and discrete aspects \cite{Geo1} (i.e., jammers placements), but  they are necessarily to be solved separately.
Moreover, all the above works do not consider the issues of WEH and EE for friendly jammers to
prolong their lifetime, which are our main focuses.  Another important line of this work is the study
of  competitive ratio of the admitted friendly jammers with instant power allocation in the distributed setting for secure wireless communications for the first time. We note that existing related works on distributed scenarios  only studied the competitive ratio for   capacity maximization \cite{2009} and online admission control  \cite{Online13}   in classic wireless communications.

The rest of this paper is organized as follows: Section II describes our system model.
 Section III  proposes the JAM-NET lifetime maximization problem, we analyze its learning performance in
 several typical and practical implementations. Section IV focuses on the distributed online jammer
 placement and power allocation problem with competitive analysis. We extend our results to
 the requests with duration and multi-channel scenarios in Section V.   Simulation  results are presented in Section VI.
 The paper is concluded in Section VII.

\section{System Models}

\subsection{Environment Model}
We consider a Storage/Fence environment model in which legitimate communication takes place within an enclosure specified by one
 or more polygonal regions $\mathfrak{S} \subset {\mathbb{R}^3}$, called the storage. We do not assume any knowledge of the locations
of  communication links  in $\mathfrak{S}$, but we do assume some properties of legitimate communication described below. At first, the
legitimate transmitters and receivers can be located at any point $p_s \in \mathfrak{S}$. Further, there exists a \emph{controlled
region}, $\mathfrak{C} \subset {\mathbb{R}^3}$, like the band region that contains $\mathfrak{S}$ in Fig. 1, where there is
no eavesdropper able to be within the interior of $\mathfrak{C}$. The boundary $\partial\mathfrak{C}$ is referred to as the fence
$\mathfrak{F}$. Outside of $\mathfrak{C}$ is the \emph{uncontrolled region} $\mathfrak{U}$. We assume that there is no hole  in $\mathfrak{C}$, which
is a union of simply connected regions; otherwise it can be divided into different regions of storages and fences. We are necessary to
place the jammers distributively in a region $\mathfrak{A}$ called the \emph{allowable region}. The allowable region permits
us to place the jammers at potential restrictions on locations that belong to the region of $\mathfrak{C} / \mathfrak{S}$, e.g., the
guarded distance for secure commutations or locations that are easily reached for maintenance purpose.

\subsection{Wireless Power Transfer and Energy Harvesting Model}
\begin{figure}
\vspace{-.2cm}
\centering
\includegraphics[scale=.76]{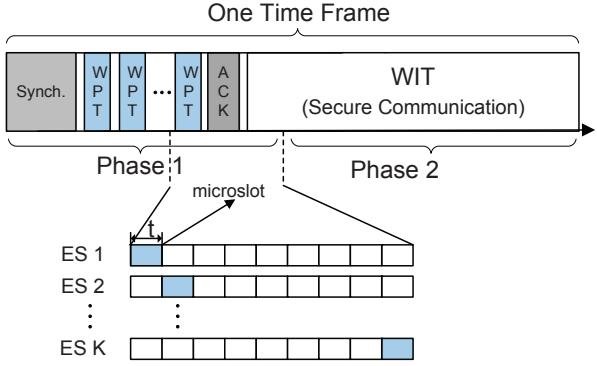}
\caption{Time frame divided into two phases. The phase one is for scheduling WPT among multiple energy sources (ESs) and the
phase two is for WIT for secure communications.}
\label{fig:digraph}
\vspace{-.5cm}
\end{figure}

We consider the WPT problem with a set of energy sources (ESs) $\mathcal{X}= \{1,2,...,k,...,K\}$ as
 legitimate  transmitters (Txs) from the legitimate communication to power a set of friendly jammers $\mathcal{A} = \{1,2,...j,...,,J\}$ with
maximal cardinality $J$ under all observed time frames. The ESs are assumed to
have no restriction to supply the energy
request to each jammer, but the served energy must be used energy-efficiently. Since ESs are geographically randomly distributed and the  stochastic placements of jammers, it is desirable to
design an online WPT scheduling algorithm (scheduler) from ESs to jammers for WEH  over time. We employ the time-switching protocol as illustrated
in Fig. 2, where a  communication frame is divided into the two phases. The phase of WPT is divided into three sub-phases:
i) perform a time synchronization among all ESs to avoid con-channel interference; ii) using the
idea of time-division multiplexing, the time block of second phase is further divided into micro slots. Online algorithms are
necessary, which schedule the WEH from a single ES $k$ a single jammer $j$ at each microslot (short as 'slot') $t$ based on the observed  CSI  and energy demands from  jammers
to ESs (detailed below); and iii) send an ACK to indicate the completion of the WPT phase. The jammers can also initiate the next frame for WEH at the end of last frame when their remaining energy-level is lower than some threshold. Hence the WIT phase is not
necessarily to be launched at each time frame.

Let the equivalent complex baseband channel from a ES $k$ to a jammer $j$ is denoted by $g^t_{k,j}(v)$, where $v$ denotes the fading
state of the CSI with the instant channel power gain $h^t_{k,j}(v)=|g^t_{k,j}(v)|^2$ at a micro slot $t$. For each
fading state $v$, the portion of signal power split to secure communication is denoted by $\alpha(v)$ with $0\le \alpha(v) \le1$, and that
to WEH as $1-\alpha(v)$. In general the $\alpha(v)$ can be adjusted over different fading state over time. For the WPT phase, the
scheduling strategies contains the following two cases, \textbf{Case \textup{I}}: $h^t_{k,j}(v)$ is perfectly known at ESs for each fading state
$v$, referred to as the \emph{known CSI at ESs}, which is a simple model but not realistic; \textbf{Case \textup{II}}: $h^t_{k,j}(v)$ is unknown at at the ESs (Txs) for all the fading state $v$, referred to as \emph{unknown CSI at ESs}, which is practical problem as the
passive jammers usually are not capable to communicate with ESs to estimate the CSI. In this case, we need to
estimate the mean values of $h^t_{k,j}(v)$ by online learning algorithms;  Given $\alpha(v)$, the harvested energy (HE) at each slot $t$ (normalized
to unitary) at each fading state $v$ can be expressed as
 \begin{IEEEeqnarray*}{l}
Q_{k,j}(v)=\xi(1-\alpha(v))h^t_{k,j}(v)p^t_{k,j}(v),
\IEEEyesnumber \label{eq:WEH1}
 \end{IEEEeqnarray*}
where $\xi$ is a constant coefficient that accounts for the loss in the energy transducer for converting the HE to electrical energy to be stored, and $p^t_{k,j}(v)$ denotes the scheduled energy from ES $k$ to $j$ at slot $t$. Denoted
each jammer has an energy demand $c_{X_{t,j}}$ at slot $t$, we have the scheduled energy request from ES $X_{t}$ at $t$ such that
 \begin{IEEEeqnarray*}{l}
p^t_{k,j}(v)=c_{X_{t,j}},  X_{t} = k.
\IEEEyesnumber \label{eq:Cost1}
 \end{IEEEeqnarray*}
  Then, the expected HE at each jammer is then given by
 \begin{IEEEeqnarray*}{l}
\mu_{k,j}=\mathbb{E}_v[Q_{k,j}(v)], \forall k \in \mathcal{X}, j \in \mathcal{J}.
\IEEEyesnumber \label{eq:WEH1}
 \end{IEEEeqnarray*}


Based on the time-division scheme, at each slot $t$ for WEH, only one ES is activated for WPT according to its distance to the boundary of storage $\mathfrak{S}$, based on which the appearance of an ES $X_t$ is normalized as an independently with identical distributed $\mathbb{P}\{X_t = k\} = \pi_k,k\in\mathcal{X}$. In this scheme, each scheduled jammer $j\in\mathcal{A}$ generates a non-negative HE (as reward)  $Y_{j,t}$. W.l.o.g., under a given ES $X_t = k$, the HE $Y_{j,t}$'s are independent random variables in $[0,1]$, where the
better channel gain $h^t_{k,j}(v)$ (e.g., schedule the nearer jammer) and the larger  energy demand provide better HE at jammer $j$. But, the conditional expectation $\mathbb{E}[Y_{j,t}|X_t=k]=u_{k,j}$ is unknown to the scheduler. Moreover, an energy demand is realized if jammer $j$ is scheduled under ES $k$. We  consider fixed and known energy demands in this paper, where the $c_{k,j}>0$ when jammer $j$ is served by ES $k$.

 We can formulate the problem as a constrained context-bandit (CMAB) \cite{Bubeck12} online learning problem. Similar to traditional contextual bandits, the ES $X_t$ as a \emph{context} is observable at the beginning $t$, while only the HE of the jammer taken by the scheduler is revealed at the end of slot $t$. Specifically, at the beginning of slot $t$, the scheduler observes the context $X_t$ and takes a jammer $A_t$ from $\{0\}\bigcup\mathcal{A}$, where ``0" represents a \emph{dummy} jammer that the scheduler  skips the current context. Let $Y_t$ and $Z_t$ be the HE  and energy demand received for the scheduler  in slot $t$, respectively. If the scheduler  takes a jammer $A_t = j>0 $, then the HE  is $Y_t=Y_{j,t}$ and the energy demand is $Z_t= c_{X_{t,j}}$. Otherwise, if the scheduler  takes the dummy jammer $A_t=0$, neither HE  nor energy demand is incurred, i.e., $Y_t=0$ and $Z_t=0$. We focus on the CMAB in this work with a known time-horizon $T$ and limited energy budget $\bar{E}$ for the goal of EE in WPT, where the process ends when the scheduler  runs out of the energy budget or at the end of time $T$.

 Formally, an online learning algorithm $\digamma$ is a function that maps the historical observations $\mathcal{H}_{t-1}=(X_1,A_1,Y;X_2,A_2,Y_2;...;X_{t-1},A_{t-1},Y_{t-1})$ and the current ES $X_t$ to a jammer $A_t\in\{0\}\bigcup\mathcal{A}$. The objective of the algorithm is to maximize the expected total HEs  $U_{\digamma}(T,\bar{E})$ for a given   $T$ and an energy budget $\bar{E}$, i.e.,
$$\textup{maximize} \ U_{\digamma}(T,\bar{E})= \mathbb{E}_{\digamma}[\sum_{t=1}^TY_t]$$
$$\textup{subject to}\  \sum_{t=1}^TZ_t \leq \bar{E},$$
where the expectation is taken over the distributions of ESs and energy rewards. Note that we consider a ``hard" energy budget constraint, i.e., the total energy demands should not be greater than $\bar{E}$ under any realization.

We measure the performance of the algorithm $\digamma$ by comparing it with the optimal one, which is the optimal algorithm with known statistics, including the knowledge of $\pi_k$'s, and $u_{k,j}$'s and $c_{k,j}$'s. Let $U^*(T,\bar{E})$ be the expected total HE for the jammers' networks (JAM-NET) obtained by an offline optimal algorithm (with hindsight full knowledge). Then, the we define the term ``\emph{regret}" of the algorithm $\Gamma$  as
$$R_{\Gamma}(T,\bar{E}) = U^*(T,\bar{E})-U_{\Gamma}(T,\bar{E}).$$
The objective of $\digamma$ is then to minimize the regret. We are interested in the asymptotic performance,  where the time-horizon 	$T$ and the energy budget $\bar{E}$ grow to infinity proportionally, i.e., with a fixed ratio $\rho= \bar{E}/T$. The coordination process of ESs to use up the total
budget $\bar{E}$ is discussed in Section III.

\subsection{Secure Communication Model}
The communication model is similar to that of \cite{Geo1, Geo2, Geo3}. W.l.o.g., we assume that the transmission power $\tilde P$ is the
same for all legitimate transmitters and  the receiving power $\bar P$ is nearly the same for all receivers within $\mathfrak{S}$. On the other hand,
the heard signal at eavesdropper suffers path loss with the path-loss exponent $\gamma$. Formally, for an eavesdropper $p_e$ listens
 to a transmitter $p_s \in \mathfrak{S}$; the received power is $\tilde P d_{p_sp_e}^{-\gamma}$, where $\gamma$ typically in range $[2,6]$ and
 $d_{pq}=||p-q||$ is the Euclidean distance between $p$ and $q$.

We assume that it is co-channel interference free among legitimate
communications and we use the \emph{Signal-to-interference Ratio} (\textup{SIR}) as physical model that only jamming
signals cause interference. Formally, for a legitimate receiver $p_s$ in $\mathfrak{S}$,
it is the ratio of the transmitted signal power to the total interference contributed by the jammers,
 \begin{IEEEeqnarray*}{l}
\bm{\textup{{SIR}}}(J,{p_s}) = \frac{{\bar P  }}{{\sum\nolimits_{j \in J} {{P_j}{d_{jp_s}^{ - \gamma }}} }},
\IEEEyesnumber \label{eq:SIR1}
 \end{IEEEeqnarray*}
where $P_j$ is the transmission power of jammer $j$. Similarly, for an eavesdropper $p_e$, the transmission signal from the
storage suffers path loss. In the case the jammer is passive, we make an observation point that the maximal signal power received
from a transmitter at the nearest location of $p_e$ on $\mathfrak{S}$ is $s(p_e)$, and we define
 \begin{IEEEeqnarray*}{l}
\bm{\textup{{SIR}}}(J,{p_e}) = \frac{{\tilde P{  d_{s({p_e}){p_e}}^{ - \gamma }}}}{{\sum\nolimits_{j \in J} {{P_j}{  d_{j{p_e}}  ^{ - \gamma }}} }}.
\IEEEyesnumber \label{eq:SIR2}
 \end{IEEEeqnarray*}
In \cite{Geo1, Geo2, Geo3}, the authors use a simple truth that the total interference at a location $p$ is usually
dominated by the interference from the nearest jammer to $p$ due to the received power decreasing
exponentially with distance. However, our model gives religious interference analysis that contributed by all transmitters.

To facilitate successful receptions of legitimate transmissions, we require that $\bm{\textup{{SIR}}}(J,{p_s}) > \delta_s$, for
all points $p_s \in \mathfrak{S}$ and some specified value of $\delta_s$. Similarly, to make the eavesdroppers unable to receive secure
messages, we require that $\bm{\textup{{SIR}}}(J,{p_e}) < \delta_e$, for points $p_e \in {\mathbb{R}^2} / \mathfrak{C}$ and some parameter
$\delta_e$. In summary, the set of friendly jammers need to satisfy the following constraints:
 \begin{IEEEeqnarray*}{l}
  \begin{array}{l}
\bm{\textup{{SIR}}}(J,{p_s}) > \delta_s, \forall p_s \in \mathfrak{S}, \IEEEyesnumber \label{eq:Cst1} \\
\bm{\textup{{SIR}}}(J,{p_e}) < \delta_e, \forall p_e \in \mathbb{R}^2 / \mathfrak{C},
 \end{array}\IEEEyesnumber \label{eq:Cst2}
 \end{IEEEeqnarray*}
Finally, as indicated in \cite{Geo2}, jamming the eavesdroppers at the fence $\mathfrak{F}$ is sufficient to
ensure that eavesdroppers located outside the fence are also jammed successfully; ensuring the
legitimate receiver on the boundary of $\mathfrak{F}$ is not jammed is sufficient to guarantee that receivers inside
 $\mathfrak{F}$  are not jammed.

\section{JAM-NET lifetime maximization: Constrained Contextual Bandits with Heterogeneous Costs}
We name the energy harvest problem of JAM-NET as the \emph{JAM-NET lifetime maximization} problem. We discuss the design of algorithms in heterogeneous-energy-demand systems where energy demand $c_{k,j}$ depends on $k$ and $j$. We summary the main results here first and then go into details. 

To schedule the WPT from multiple ESs to jammers, we show that significant complexity incurs when making decisions under those
"medium-quality" ESs other the best and worst ESs. Because the scheduler needs to balance between the instantaneous expected harvest
energy rewards and the future rewards, which is very difficult due the coupling effect introduced by the time and energy budgets constraints. Thus,
we resort to approximations by relaxing the hard time and energy budget constraints to an average energy  budget constraint. Now, the
problem becomes to maximize the expected rewards (of total harvest energy) with average budget constraint $\bar E/T$. However, the problem
of using the fixed average budget $\bar E/T$ and not taking the remaining time $\tau$ and remaining energy  budget $b_\tau$ into account would not
explore the dynamic structure of the problem, which will lead to suboptimal solutions. Hence, we propose an Adaptive Integer Linear Program (AILP) that replaces  the fixed average budget $\bar E/T$ by the \emph{remaining average budget}, i.e., $b_\tau/\tau$. We indicate that the
performance analysis is non-trivial for the AILP, although the intuition behind the formulation is quite natural. Using concentration inequalities, we can
show that $b_\tau/\tau$ under AILP concentrates near the average budget $\bar E/T$ with high probability. This proves that the proposed AILP algorithms
achieves an expected total HE only within a constant to the optimum expect for certain boundary cases.

Next, gaining the insight from the case of the known statistics of the expected HE for each jammer, we extend our
analysis to that the expected HE as rewards for jammers are unknown.  During the analysis, we find
that the AILP algorithm only require the ranking of the expected rewards rather than their actual values. Inspired by this, we combine
the famous upper-confidence-bound (UCB) algorithm \cite{Finite02} and show that,  for the general heterogeneous-energy-demand systems, our
proposed $\epsilon$-ESs-JamNet-UCB-AILP algorithm achieves $O(log(t))$ regret that is learning-rate optimal except for certain boundary
case, where it achieves $O(\sqrt{t})$ regret.


\subsection{Approximation of the Offline Optimal Algorithm: Known CSI at ESs}
In this subsection, we first study the case with known statistics. We consider the case where the energy demand for each jammer$k$ under context $j$ is fixed at $c_{k,j}$, which may be different for different $k$ and $j$.

With known statistics, the scheduler  knows the context (ESs) distribution $\pi_k$'s, the energy demands $c_{k,j}$'s, and the expected rewards of HEs $u_{k,j}$'s.  With heterogeneous energy demands, the quality of the WEH process a jammer $j$ under a context $k$ is roughly captured by its \emph{normalized expected HE}, defined as $\eta_{k,j}=u_{k,j}/c_{k,j}$. However, the scheduler  cannot only focus on the ``best" jammer to serve, i.e., $j_k^*=argmax_{j\in\mathcal{A}}\eta_{k,j}$, for context ES $k$. This is because there may exist another jammer$j'$ such that $\eta_{k,j'}<\eta_{k,j_{k}^*}$, but $u_{k,j'}>u_{k,j_{k}^*}$ (and surely, $c_{k,j'}>c_{k,j_{k}^*}$). If the energy budget allocated to each ES $k$ out of the total is sufficient,
then the scheduler  may take jammer $j'$ to maximize the expected reward of HE. Therefore, the ALP algorithm in this case needs to decide the probability to take jammer $j$ under ES $k$, by solving an ILP problem with an additional constraint that only on jammer can be taken under each ES.   We can show that ALP achieves $O(1)$ regret in non-boundary cases, and $O(\sqrt{T})$ regret in boundary cases. We note that the regret analysis of ALP in this case is much more difficult due to the additional constraint that couples all jammers under each ES.

In this case, the scheduler  needs to consider all
jammers under each ES. Let $p_{k,j}$ be the probability that jammer$j$ is taken under ES $k$.
We define the following ILP problem:
$$
(\mathcal{LP}_{T, \bar{E}}')\quad \textup{maximize}  \quad\sum_{k=1}^K\pi_k\sum_{j=1}^Jp_{k,j}u_{k,j},\eqno{(7)}
$$\IEEEyesnumber
\quad subject to:
\begin{subequations}\label{eq:CstILP1}
\begin{numcases}{}
\sum_{k=1}^K\pi_k\sum_{j=1}^Jp_{k,j}c_{k,j}\leq \bar{E}/T,\label{test 1}\\
\sum_{j=1}^Jp_{k,j}\leq 1, \forall k,\label{test 2}\\
p_{k,j}\in[0,1].
\end{numcases}
\end{subequations}
\indent The above LP problem $\mathcal{LP}_{T,\bar{E}}'$ can be solved efficiently by optimization tools. Let $\hat{v}(p)$ be the maximum value of  $\mathcal{LP}_{T,\bar{E}}'$. Similar to Lemma 1, we can show that $T\hat{v}(p)$ is an upper bound of the expected total HE, i.e., $T\hat{v}(p)\leq U^*(T,\bar{E})$.\\
To obtain insight from the solution of $\mathcal{LP}_{T,\bar{E}}'$, we derive an explicit representation for the solution by analyzing the structure of $\mathcal{LP}_{T,\bar{E}}'$. Note that there are two types of (non-trivial) constraints in $\mathcal{LP}_{T,\bar{E}}'$, one is the ``inter-ES" energy budget constraint (\ref{test 1}), the other is the ``intra-ES" constraint (\ref{test 2}). These constraints can be decoupled by first allocating energy budget for each context, and then solving a subproblem with the allocated energy budget constraint for each ES. Specifically, let $\rho_k$ be the energy budget allocated to ES $k$, then $\mathcal{LP}_{T,\bar{E}}'$ can be decomposed as follows:
$$\textup{maximize}\quad\sum_{k=1}^K\pi_k\hat{v_k}(\rho_k),$$
$$\quad \ \textup{subject to: } \sum_{k=1}^K\pi_k\rho_k\leq \bar{E}/T$$
where
\begin{subequations}
\begin{numcases}{}
(\mathcal{SP}_k) \ v_k(\rho_k)= \textup{maximize} \ \sum_{j=1}^Jp_{k,j},\\
\textup{subject to}\quad\sum_{j=1}^Jp_{k,j}c_{k,j}\leq\rho_k,\\
\quad\quad\quad\quad\quad\sum_{j=1}^Jp_{k,j}\leq 1,p_{k,j}\in[0,1].
\end{numcases}
\end{subequations}
\indent Next, by analyzing sub-problem $\mathcal{SP}_k$, we show that some jammers can be deleted without affecting the performance, i.e., the probability is 0 in the optimal solution.

\begin{lemma}\label{lemma11}
For any given $\rho_k\geq 0$, there exists an optimal solution of $\mathcal{SP}_k$, i.e., $\textbf{p}_k^*=(p_{k,1}^*, p_{k,2}^*,..., p_{k,J}^*)$, satifies:
\begin{enumerate}[(1)]
\item For $j_1$, if there exists another jammer $j_2$, such that $\eta_{k,j_1}\leq\eta_{j_2}$ and $u_{j_1}\leq u_{k,j_2}$, then $p_{k,j_1}^*=0$;
\item For $k_1$, if there exists two jammers $j_2$ and $j_3$, such that $\eta_{k,j_2}\leq\eta_{k,j_1}\leq\eta_{k,j_3}$, $u_{k,j_2}\leq u_{k,j_1}\leq u_{k,j_3}$, and $\frac{u_{k,j_1}-u_{k,j_3}}{c_{k,j_1}-c_{k,j_3}}\leq\frac{u_{k,j_2}-u_{k,j_3}}{c_{k,j_2}-c_{k,j_3}}$, then  $p_{k,j_1}^*=0$.
\end{enumerate}
\end{lemma}
Intuitively, the first part of Lemma \ref{lemma11} shows that if a jammer has small normalized and original expected HE, then it can be removed. The second part of Lemma \ref{lemma11} shows that if a jammer has small normalized expected HE and medium original expected HE, but the increasing rate is smaller than another jammer with larger expected HE, then it can also be removed.

\begin{proof}
The key idea of this proof is that, if the conditions is satisfied, and there is a feasible solution $\textbf{p}_k = (p_{k,1},p_{k,2},...,p_{k,J})$ such that $p_{k,j_1}>0$, then we can construct another feasible solution $\textbf{p}'_k$ such that $p'_{k,j_1}=0$, without reducing the objective value $v_k(\rho_k)$.\\
\indent We first prove part (1). Under the conditions of part (1), if  $\textbf{p}_k$ is a feasible solution of $\mathcal{SP}_k$ with $p_{k,j_1}>0$, then consider another solution $\textbf{p}'_k$, where $p'_{k,j}=p_{k,j}$ for $j\notin\{j_1,j_2\}$, $p'_{k,j}= 0$, and $p'_{k,j_2}=p_{k,j_2}+p_{k,j_1} \textup{min}\{\frac{c_{k,j_1}}{c_k,j_2},1\}$.  Then, we can verify that $\textbf{p}'_k$ is a feasible solution of ($\mathcal{SP}_k$), and the objective value under $\textbf{p}'_k$ is no less than that under $\textbf{p}_k$.

For the second part, if the conditions are satisfied and $p_{k,j_1}>0$, then we construct a new solution $\textbf{p}'_k$ by re-allocating the energy budget consumed by jammer $j_1$ to jammer $j_2$ and $j_3$, without violating the constraints. Specifically, we set the probability the same as the original solution for other jammers, i.e., $p'_{k,j} = p_{k,j}$ for $j\notin\{j_1,j_2,j_3\}$, and set $p'_{k,j_1} = 0$ for jammer $j_1$. For $j_2$ and $j_3$, to maximize the objective function, we would like to allocate as much energy budget as possible to $j_3$ unless there is remaining energy budget. Therefore, we set $p'_{k,j_2} = p_{k,j_2}$ and $p'_{k,j_3} = p_{k,j_3} + \frac{p_{k,j_1}c_{k,j_1}}{c_{k,j_3}}$, if $\sum_{j\neq j_1} p_{k,j} + \frac{p_{k,j_1}c_{k,j_1}}{c_{k,j_3}} \leq 1$; or, $p'_{k,j_2} = p_{k,j_2} + \frac{p_{k,j_1}c_{k,j_1}-(1-\sum_{j\neq j_1}p_{k,j})c_{k,j_3}}{c_{k,j_2}-c_{k,j_3}}$ and $p'_{k,j_3} = p_{k,j_3} + \frac{(1-\sum_{j\neq j_1}p_{k,j})c_{k,j_2}-p_{k,j_1}c_{k,j_1}}{c_{k,j_2}-c_{k,j_3}}$, if $\sum_{j\neq j_1} p_{k,j} +\frac{p_{k,j_1}c_{k,j_1}}{c_{k,j_3}}>1$. We can verify that $\textbf{p}_k$ satisfies the constraints of ($\mathcal{SP}_k$) but the objective value is no less than that under $\textbf{p}_k$
\end{proof}
With Lemma \ref{lemma11}, the scheduler  can ignore some jammers that will obviously be allocated with zero probability under a given context $k$. We call the set of the remaining jammers as \emph{candidate set} for context $k$, denoted as $\mathcal{A}_k$. We propose an algorithm to construct the candidate jammer set for context $k$, as shown in Algorithm \ref{algorithm 4}.

\begin{algorithm}[h]
\caption{Find Candidate Jammer Set for ES $k$}\label{algorithm 4}
\begin{algorithmic}[1]
\REQUIRE $c_{k,j}$'s, $u_{k,j}$'s, for all $1 \leq j \leq J$;\
\ENSURE $\mathcal{A}_k$;\
\STATE Calculate normalized HE as rewards: $\eta_{k,j} = u_{k,j}/c_{k,j}$;\
\STATE Sort jammers in descending order of their normalized rewards of HEs:$$\eta_{k,j_1}\geq\eta_{k,j_2}\geq...\geq\eta_{k,j_J}.$$
\FOR {$a=2$ \TO $J$}
\IF{$\exists a'<a$ such that $u_{k,j_a}\leq u_{k,j_{a'}}$}
\STATE $\mathcal{A}_k=\mathcal{A}_k/\{j_a\}$;
\ENDIF
\ENDFOR
\STATE $a=1$;
\WHILE{$a\leq J-1$}
\STATE  Find the jammer with highest increasing rate: $$a^*=\textup{argmax}_{a':a'>a,j_{a'}\in\mathcal{A}_k} \frac{u_{k,j_{a'}}-u_{k,j_a}}{c_{k,j_{a'}}-c_{k,j_a}}.$$ \
\STATE Remove the jammers in between:$$\mathcal{A}_k=\mathcal{A}_k/\{j_{a'}:a<a'<a^*\}.$$\
\STATE Move to the next candidate jammer: $a=a^*$;
\ENDWHILE
\end{algorithmic}
\end{algorithm}

 For ES $k$, assume that the candidate jammer set $\mathcal{A}_k = \{j_{k,1}, j_{k,2},..., j_{k,J_k}\}$ has been sorted in descending order of their normalized energy rewards, i.e., $\eta_{k,j_{k,1}}\leq\eta_{k,j_{k,2}}\leq...\leq\eta_{k,j_{k,J_k}}$. From Algorithm \ref{algorithm 4}, we know that $ u_{k,j_{k,1}}<u_{k,j_{k,2}}<...< u_{k,j_{k,J_k}}$, and $c_{k,j_{k,1}}<c_{k,j_{k,2}}<...<c_{k,j_{k,J_k}}$.\\
\indent The scheduler  now only needs to consider the jammers in the candidate set $\mathcal{A}_k$. To decouple the ``intra-ES" constraint (\ref{test 2}), we introduce the following transformation:
\begin{equation}
p_{k,j_{k,a}}= \left\{ \begin{array}{ll}
\widetilde{p}_{k,j_{k,a}}-\widetilde{p}_{k,j_{k,a+1}}, & \textrm{if $1\leq a\leq J_{k}-1$, }\\
\widetilde{p}_{k,j_{k,J_k}}, & \textrm{if $a = J_k$,}\\
\end{array} \right.
\end{equation}
where $\widetilde{p}_{k,j_{k,a}}\in[0,1]$, and $\widetilde{p}_{k,j_{k,a}}\leq\widetilde{p}_{k,j_{k,a+1}}$ for $1\leq a\leq J_k-1$. Substituting the transformations into ($\mathcal{SP}_k$) and reformulate it as
$$
(\widetilde{\mathcal{SP}_k}) \ \textup{maximize}\sum_{a=1}^{J_k}\widetilde{p}_{k,j_{k,a}}\widetilde{u}_{k,j_{k,a}},
$$
subject to:
\begin{subequations}\label{test 3}
\begin{numcases}{}
\sum_{a=1}^{J_k}\widetilde{p}_{k,j_{k,a}}\widetilde{c}_{k,j_{k,a}}\leq \rho_k,\\
\widetilde{p}_{k,j_{k,a}}\leq\widetilde{p}_{k,j_{k,a+1}}, 1\leq a\leq J_k-1,\\
\widetilde{p}_{k,j_{k,a}}\in[0,1], \forall a,
\end{numcases}
\end{subequations}
\noindent where
\begin{equation}
\widetilde{u}_{k,j_{k,a}}= \left\{ \begin{array}{ll}
u_{k,j_{k,1}}, & \textrm{if  $a=1$, }\\
u_{k,j_{k,a}}-u_{k,j_{k,a-1}}, & \textrm{if $2\leq a \leq J_k$,}\\
\end{array} \right.
\end{equation}
\begin{equation}
\widetilde{c}_{k,j_{k,a}}= \left\{ \begin{array}{ll}
c_{k,j_{k,1}}, & \textrm{if  $a=1$, }\\
c_{k,j_{k,a}}-c_{k,j_{k,a-1}}, & \textrm{if $2\leq a \leq J_k$,}\\
\end{array} \right.
\end{equation}
\indent Next, we show that the constraint(\ref{test 3}) can indeed be removed. For each $j_{k,a}$, we can view $\widetilde{c}_{k,j_{k,a}}$ and $\widetilde{u}_{k,j_{k,a}}$ as the energy demand and expected HE of a virtual jammer. Let $\widetilde{\eta}_{k,j_{k,a}}= \widetilde{u}_{k,j_{k,a}}/\widetilde{c}_{k,j_{k,a}}$  be the normalized expected HE of virtual jammer $j_{k,a}$. For $a=1$, using $\frac{u_{k,j_{k,1}}}{c_{k,j_{k,1}}}\geq\frac{u_{k,j_{k,2}}}{c_{k,j_{k,2}}}$, we can show that $\widetilde{\eta}_{k,j_{k,1}}\geq\widetilde{\eta}_{k,j_{k,2}}$. For $2\leq a \leq J_k-1$ using $\frac{u_{k,j_{k,a}}-u_{k,j_{k,a-1}}}{c_{k,j_{k,a}}-c_{k,j_{k,a-1}}}\geq\frac{u_{k,j_{k,a+1}}-u_{k,j_{k,a-1}}}{c_{k,j_{k,a+1}}-c_{k,j_{k,a-1}}}$, we can show that $\widetilde{\eta}_{k,j_{k,a}}\geq\widetilde{\eta}_{k,j_{k,a+1}}$. In other words, we can verify that $\widetilde{\eta}_{k,j_{k,1}}\geq\widetilde{\eta}_{k,j_{k,2}}\geq...\geq\widetilde{\eta}_{k,j_{k,J_k}}$. Thus, without constraint (\ref{test 3}), the optimal solution $\widetilde{\textbf{p}}_{k}=[\widetilde{p}^*_{k,j_1}, \widetilde{p}^*_{k,j_2},...,\widetilde{p}^*_{k,j_{J_k}}]$ automatically satisfied $\widetilde{p}^*_{k,j_1}\geq \widetilde{p}^*_{k,j_2}\geq\widetilde{p}^*_{k,j_{J_k}}$. Hence, we can remove the constraint (\ref{test 3}), and thus decouple the probability constraint under a ES.

We can thus rewrite the global ILP problem using the above transformations
$$
(\widetilde{\mathcal{LP}}'_{T,\bar{E}}) \ \textup{maximize}\sum_{k=1}^{K}\sum_{a=1}^{J_k}\pi_k\widetilde{p}_{k,j_{k,a}}\widetilde{u}_{k,j_{k,a}},$$
\quad\quad\quad\quad subject to  $\quad \sum_{k=1}^{K}\sum_{a=1}^{J_k}\pi_k\widetilde{p}_{k,j_{k,a}}\widetilde{c}_{k,j_{k,a}}\leq \frac{\bar{E}}{T},$
$$\quad\quad\quad\quad\quad\quad\quad \quad  \widetilde{p}_{k,j_{k,a}}\in[0,1],\forall k,  1\leq a\leq J_k.$$

 The solution of $\widetilde{\mathcal{LP}}'_{T,\bar{E}}$ follows a threshold structure. We sort all ES-(virtual-)jammer pairs ($k,j_a$) in descending order of their normalized expected HEs. Let $k^{(i)}$, $j^{(i)}$ be the context index and jammer index of the $i$-th pair, respectively. Namely, $\widetilde{\eta}_{k^{(1)},j^{(1)}}\geq\widetilde{\eta}_{k^{(2)},j^{(2)}}\geq...\geq\widetilde{\eta}_{k^{(M)},j^{(M)}}$, where $M = \sum_{k=1}^K J_k$ is the total number of candidate jammers for all ESs. Define a threshold corresponding to $\rho = \bar{E}/T$,
\begin{equation}
\widetilde{i}(\rho)= \textup{max}\{i:\sum_{i'=1}^i\pi_{k^(i')}\widetilde{c}_{k^{(i')},j^{(i')}}\leq \rho\},
\end{equation}
\noindent where $\rho=\bar{E}/T$ is the average energy budget. We can verify that the following solution is optimal for $\widetilde{\mathcal{LP}}'_{T,\bar{E}}$:
\begin{equation}
\widetilde{p}_{k^{(i)},j_{(i)}}= \left\{ \begin{array}{ll}
1, & \textrm{if  $1\leq i\leq \widetilde{i}(\rho)$, }\\
\frac{\rho-\sum_{i'=1}^{\widetilde{i}(\rho)}\pi_{k^(i')}\widetilde{c}_{k^(i'),j^(i')}}{\pi_{k^{(\widetilde{i}(\rho))}}\widetilde{c}_{k^{(\widetilde{i}(\rho)+1)},j^{(\widetilde{i}(\rho)+1)}}}, & \textrm{if $i = \widetilde{i}(\rho)+1$,}\\
0, & \textrm{if $i > \widetilde{i}(\rho)+1$.}\\
\end{array} \right.
\end{equation}
\noindent Then, the optimal solution of $\widetilde{\mathcal{LP}}'_{T,\bar{E}}$ can be calculated using the reverse transformation from $\widetilde{p}_{k,j}(\rho)$'s to $p_{k,j}(\rho)$'s.
\subsubsection{ALP Algorithm}
Obviously, the ALP algorithm replaces the average constraint $\bar{E}/T$ in $ \widetilde{\mathcal{LP}}'_{T,\bar{E}}$ with the average remaining energy budget $b_\tau/\tau$, and obtains probability $p_{k,j}(b_\tau/\tau))$. Under context $k$, the ALP algorithm take jammer $j$ with probability $p_{k,j}(b_\tau/\tau)$.  Note that the remaining energy budget $b_\tau$ does not follow any classic distribution in heterogeneous-energy-demand systems. However, by using the method of averaged bounded differerces\cite{23}, we show that there has the concentration property holds.

\begin{lemma}\label{lemma 12}
For $0<\delta <1$, there exists a positive number $\mathcal{K}$, such that under the $ALP$ algorithm, the remaining energy budget $b_\tau$ satisfies
$$\mathbb{P}\{b_\tau >(\rho+\delta)\tau\}\leq e^{-\mathcal{K}\delta^2\tau},$$
$$\mathbb{P}\{b_\tau <(\rho-\delta)\tau\}\leq e^{-\mathcal{K}\delta^2\tau}.$$
\end{lemma}

\begin{proof}
We prove the lemma using the method of averaged bounded differences\cite{23}. The process is similar to {{Section 7.1}} in \cite{23}, except that we consider the remaining energy budget and the successive differences of the remaining energy budget are bounded by $c_{\textup{max}}$.\\
\indent  Specifically, let $\widetilde{c}_{t'}$, $1\leq t'\leq T$ be the energy budget consumed under $ALP$, and let $\widetilde{\textbf{c}}_{t'}= (\widetilde{c}_{1},\widetilde{c}_{2},...,\widetilde{c}_{t'})$. Then the remaining energy budget at slot $t$ (the remaining time $\tau=T-t+1$), i.e., $b_{T-t+1}$ is a function of $\widetilde{\textbf{c}}_t$. We note that under ALP, the expectation of the ratio between the remain energy budget and the remaining time does not change, i.e., for any $b\leq\sum_{j=1}\pi_{k}c_j^*$(here $c_j^*=\textup{max}_kc_{k,j}$), if $b_\tau= b$, then $\mathbb{E}[b_{\tau-1}/(\tau-1)]=b/\tau$. Thus, we can verify that for any $1\leq t'\leq t$, we have
\begin{equation}
\mathbb{E}[b_{T-t+1}|\widetilde{\textbf{c}}_{t'}]=b_{T-t'+1}-\frac{b_{T-t'+1}}{T-t'+1}(t-t').\IEEEnonumber
\end{equation}
 Note that $\Delta b =b_{T-t'+2}-b_{T-t'+1}\leq c_{\textup{max}}$ and $b_{T-t'+2}\geq -c_{\textup{max}}$, we have
$$|\mathbb{E}[b_{T-t+1}|\widetilde{\textbf{c}}_{t'}]-\mathbb{E}[b_{T-t+1}|\widetilde{\textbf{c}}_{t'-1}]|$$
$$\quad\quad\leq \textup{max}_{0\leq \Delta b\leq c_{\textup{max}}}\{|\Delta b- \frac{b_{T-t'+2}}{T-t'+2}|\}\frac{T-t+1}{T-t'+1}$$
\begin{equation}
\quad\quad\leq \frac{2c_{\textup{max}}(T-t+1)}{T-t'+1}.\IEEEnonumber
\end{equation}
Moreover,
$$\sum_{t'=1}^t[\frac{2c_{\textup{max}(T-t+1)}}{T-t'+1}]^2 = 4c_{\textup{max}}^2(T-t+1)^2\sum_{t'=1}^t\frac{1}{(T-t'+1)^2}$$
$$\quad\quad\quad\quad\quad\quad\quad =4c_{\textup{max}}^2(T-t+1)^2\sum_{\tau'=T-t+1}^T \frac{1}{(\tau')^2}$$
$$\quad\quad\quad\quad\quad\quad\quad\quad\approx 4c_{\textup{max}}^2(T-t+1)^2 \int_{T-t+1}^T \frac{1}{(\tau')^2}d\tau'$$
\begin{equation}
\quad \ \ \quad\quad\quad\quad\quad=4c_{\textup{max}}^2(T-t+1)\frac{t-1}{T}\IEEEnonumber
\end{equation}
  According to {{Theorem 5.3}} in \cite{23}, and noting $\tau = T-t+1$, $\mathbb{E}[b_\tau]=\rho\tau$, we have
\begin{equation}\mathbb{P}\{b_\tau>\mathbb{E}[b_\tau]+\delta_\tau\}\leq e^{-\frac{2T(\delta\rho\tau)^2}{4c_{\textup{max}}^2(T-t+1)(t-1)}}\leq e^{-\frac{T\delta^2\bar{E}^2\tau}{2c_{\textup{max}}^2T^2(t-1)}}\leq e^{-\frac{\delta^2\rho^2\tau}{2c^2_{\textup{max}}}},\IEEEnonumber
\end{equation}
and similarly,
\begin{equation}
\mathbb{P}\{b_\tau <\mathbb{E}[b_\tau]-\delta_\tau\}\leq e^{-\frac{\delta^2\rho^2\tau}{2c_{\textup{max}}^2}},\IEEEnonumber
\end{equation}
Choosing $\mathcal{K}=\frac{\rho^2}{2c_{\textup{max}}^2}$ concludes the proof.
\end{proof}

Then, using similar methods in {\color{red}{Section 3}}, we can show that the generalized $ALP$ algorithm achieves $O(1)$ regret in non-boundary cases, and $O(\sqrt{T})$ regret in boundary cases, where the boundaries are now defined as $Q_i= \sum_{i'=1}^i\pi_{k^{(i')}}\widetilde{c}_{k^{(i')},j^{(i')}}$.\\

\indent Next, we show that the $\epsilon$-First policy with CLT will achieve $O(logT)$ regret except for the boundary cases, where it achieves $O(\sqrt{T})$ regret. On one hand, according to Hoeffding-Chernoff bound, if all comparisons pass the confidence level test, then with probability at least $1-KJ^2T^{-2}$, the algorithm obtains the correct rank and provide a right solution for the problem $(\mathcal{LP}'_{\tau,b})$. On the other hand, because $\Delta^*>0$, from the analysis in the previous section, we know that the exploration stage will end within $O(logT)$ rounds with high probability. Therefore, the expected regret is the same as that in the case with known $\Delta_{\textup{min}}^{(\epsilon)}$.\\
\indent Lemma 2 also states that the average remaining budget $\delta$, $\tau$ stays in a neighborhood of the initial average budget $\rho$ with high probability. Hence, if the initial average budget $\rho$ is not on boundaries, i.e., the critical values under which the threshold $\widetilde{j}(\rho)$ changes, then the probability of threshold changing is bounded. Therefore, we can show that the ALP algorithm achieves a very good performance within a constant distance from the optimum, except for certain boundary cases. Specifically, for $1\leq k\leq K$, let $q_k$ be the cumulative distribution function, i.e., $q_k=\sum_{k'=1}^k\pi_{k'}$, and w.l.o.g., let $q_0=0$. The following theorem states the approximate optimality of ALP for the cases where $\rho\neq q_k(k=1,2,...,	K-1)$. We note that $k=0$ and $k=K$ are trivial cases where ALP is optimal.
\begin{theorem}
Given any fixed $\rho\in(0,1)$ satisfying $\rho\neq q_k$, $k=1,2,...,K-1$, the ALP algorithm achieves an $O(1)$ regret. Specifically,$$U^*(T,\bar{E})-U_{ALP}(T,\bar{E})\leq\frac{u_1^*-u_K}{1-e^{-2\delta^2}},$$
\noindent where $\delta-\min\{\rho-q_{\widetilde{k}},q_{\widetilde{k}+1}-\rho\}$.
\end{theorem}
\begin{proof}
The proof of this theorem uses the following two facts derived from {\color{red}{lemma 2}}: $\mathbb{E}[v(b_\tau/\tau)]=v(p)$ if the threshold $\widetilde{j}(b_\tau/\tau)=\widetilde{j}(\rho)$ for all possibel $b_\tau$'s, and the probability that $\widetilde{j}(b_\tau/\tau)\neq\widetilde{j}(\rho)$ decays exponentially. Please referred to {\color{red}{Appendix B.1}} of the supplementary material for details.\\
\indent When considering the boundary eases, we can show similarly that the ALP achieves $O(\sqrt{T})$ regret.
\end{proof}
\begin{theorem}
Given any fixed $\rho=q_j$, $k=1,2,...,K-1$, the ALP algorithm achieves $O(\sqrt{T})$ regret. Specifically, $$U^*(T,\bar{E})-U_{ALP}(T,\bar{E})\leq\Theta^{(o)}\sqrt{T}+\frac{u_1^*-u_K}{1-e^{-2{\delta'}^2}},$$
\noindent where $\Theta^{(o)}= 2(u^*_1-U^*_K)\sqrt{\rho(1-\rho)}$ and $\delta'=\min\{\rho-q_{\widetilde{j}(\rho)-1},q_{\widetilde{j}(\rho)+1}-\rho\}$.
\end{theorem}

\subsection{$\epsilon$-ESs-JamNet-UCB-AILP Algorithm: Unknown CSI at ESs}
Due to the online placement of passive jammers, the channel gain $h^t_{k,j}(v)$ is usually unknown and so that the rewards of HEs from ESs to the
jammer. Now, we consider the practical case that the expected energy rewards are unknown, and online
learning algorithms are called for, e.g. UCB \cite{Finite02}. As noticed, it is difficult to combine UCB method directly with the proposed ALP for the general
heterogeneous-energy-demand systems. However, we can narrow down to  a special and very reasonable case, when all jammers have the same energy demand under a given ES, i.e., $c_{k,j}=c_k$ for all $j$ and $k$, the normalized expected HE $\eta_{k,j}$ represents the quality of jammer $j$ under ES $k$. In this case, the candidate jammer set for each ES only contains one jammer to be scheduled, which is the jammer with the highest expected HE. Thus, the previous ALP algorithm for the known statistics case is simple. When the expected energy rewards are unknown, we can extend the UCB-AILP algorithm by managing the UCB for the normalized expected energy rewards.

When the energy demands for different jammers under the same ES are heterogeneous, it is difficult to combine AILP with the UCB method since the ALP algorithm in this case not only requires the ordering of $\eta_{k,j}$'s, but also the ordering of $u_{k,j}$'s and the ratios  $\frac{u_{k,j_1}-u_{k,j_2}}{c_{k,j_1}-c_{k,j_2}}$. We propose an $\epsilon$-ESs-JamNet-UCB-AILP Algorithm that explores and exploits separately: the scheduler  takes jammers under all ESs in the first $\epsilon(T)$ rounds to estimate the expected energy rewards, and runs $ALP$ based on the estimates in the remaining $T-\epsilon(T)$ rounds.

\begin{algorithm}[h]
\caption{$\epsilon$-ESs-JamNet-UCB-AILP}\label{algorithm 5}
\begin{algorithmic}[1]
\REQUIRE Time horizon $T$, energy budget $\bar{E}$, exploration stage length $\epsilon(T)$, and $c_{k,j}$'s, for all $k$ and $j$;\
\FOR {$t=1$ \TO $\epsilon(T)$}
\IF{$b>0$}
\STATE Take jammer $A_t=\textup{argmin}_{k\in\mathcal{A}}C_{X_t,j}$ (with random tie-breaking);\
\STATE Observe the HE  $Y_{A_t,t}$;\
\STATE Update counter $C_{X_t,A_t}=C_{X_t,A_t}+1$; update remaining energy budget $b=b-c_{X_t,A_t}$;\
\STATE Update the HE estimate:
$$\bar{u}_{X_t,A_t}=\frac{(C_{X_t,A_t}-1)\bar{u}_{X_t,A_t}+Y_{A_t,t}}{C_{X_t,A_t}}.$$
\ENDIF
\ENDFOR
\FOR {$t=\epsilon(T)+1$ \TO $T$}
\STATE Remaining time $\tau = T-t+1$;
\IF{$b>0$}
\STATE Obtain the probabilities $p_{k,j}(b/\tau)$'s by solving the problem ($\mathcal{LP}'_{\tau,b}$) with $u_{k,j}$ replaced by $\bar{u}_{k,j}$;\
\STATE Take jammer $j$ with probability $p_{X_t,j}(b/\tau)$;\
\STATE Remaining energy budget $b=b-c_{X_t,A_t}$:\
\ENDIF
\ENDFOR
\end{algorithmic}
\end{algorithm}

For the case of exposition, we assume $c_{k,j_1}\neq c_{k,j_2}$ for any $k$ and $j_1\neq j_2$\footnote{For the case with $c_{k,j_1}= c_{k,j_2}$ for some $k$ and $j_1\neq j_2$ (and $u_{k,j_1}\neq u_{k,j_2}$), we can correctly remove the suboptimal jammer with high probability by comparing their empirical energy rewards $\bar{u}_{k,j_1}=\bar{u}_{k,j_2}$}, and let $\Delta_{\textup{min}}^{(c)}=\min\limits_{\mbox{\tiny$\begin{array}{c}
k\in\mathcal{X}\\
j_1,j_2\in\{0\}\cup\mathcal{A}\end{array}$}}\{|c_{k,j_1}-c_{k,j_2}|\}$. Let $\varepsilon_{k,j_1,j_2}=\frac{u_{k,j_1}-u_{k,j_2}}{c_{k,j_1}-c_{k,j_2}}$ for $k \in\mathcal{X}$, $j_1,j_2\in\{0\}\cup\mathcal{A}$, and $j_1\neq j_2$ (recall that $u_{k,0}=0$ and $c_{k,0}=0$ for the ``dummy jammer"), $\bar{\varepsilon}_{k,j_1,j_2}$ be its estimate at the end of the exploration stage, i.e., $\bar{\varepsilon}_{k,j_1,j_2}=\frac{\bar{u}_{k,j_1}-\bar{u}_{k,j_2}}{c_{k,j_1}-c_{k,j_2}}$. Let $\Delta_{\textup{min}}^{(\varepsilon)}$ be the minimal difference between any $\varepsilon_{k_1,j_{11},j_{12}}$ and $\varepsilon_{k_2,j_{21},j_{22}}$, i.e.,
$$\Delta_{\textup{min}}^{(\varepsilon)}=\min\limits_{\mbox{\tiny$\begin{array}{c}
k_1,k_2\in\mathcal{X}\\
j_{11},j_{12},j_{21},j_{22}\in\{0\}\cup\mathcal{A}\end{array}$}}\{|\varepsilon_{k_1,j_{11},j_{12}}-\varepsilon_{k_2,j_{21},j_{22}}|\}.
$$
\indent Moreover, let $\pi_{\textup{min}}=\min_{k\in\mathcal{X}}\pi_{k}$ and let $\Delta^*=\Delta_{\textup{min}}^{(c)}\Delta_{\textup{min}}^{(\varepsilon)}$. Then, the following lemma states that under $\epsilon$-ESs-JamNet-UCB-AILP with a sufficiently large $\epsilon(T)$, the scheduler  will obtain a correct ordering of $\varepsilon_{k,j_1,j_2}$'s with high probability at the end of the exploration stage.

\begin{lemma}\label{lemma 13}
Let $0<\delta<1$. Under $\epsilon$-ESs-JamNet-UCB-AILP, if
$$\epsilon(T)=\lceil\frac{K}{(1-\delta)\pi_{\textup{min}}}+\textup{log} T \textup{max}\{\frac{1}{\delta^2},\frac{16K}{(1-\delta)\pi_{\textup{min}}(\Delta^*)^2}\}\rceil,$$
\noindent then for any contexts $k_1,k_2\in\mathcal{X}$, and jammers $j_{11},j_{12},j_{21},j_{22}\in\{0\}\cup\mathcal{A}$, if $\varepsilon_{k_1,j_{11},j_{12}}<\varepsilon_{k_2,j_{21},j_{22}}$, then at the end of the $\epsilon(T)$-th slot, we have
$$\mathbb{P}\{\bar{\varepsilon}_{k_1,j_{11},j_{12}}<\bar{\varepsilon}_{k_2,j_{21},j_{22}}\}\leq(K+4)T^{-2}.$$
\indent Moreover, the scheduler  ranks all the $\varepsilon_{k,j_1,j_2}$'s correctly with probability no less than $1-(4J+1)KT^{-2}$.\\
\end{lemma}
\begin{proof}
We first analyze the number of executions for each ES-jammer pair $(k,j)$ in the exploration stage. Let $N_k = \sum\nolimits_{t=1}^{\epsilon(T)}\mathds{1}(X_t=k)$ be the number of occurrences of ES $k$ up to slot $\epsilon(T)$. Recall that the ESs $X_t$ is activated i.i.d. in each slot by the time-division protocol. Thus, using Hoeffding-Chernoff Bound for each ES $k$, we have
\begin{flushleft}
$\mathbb{P}\{\forall k\in \mathcal{X}, N_{k}\geq(1-\delta)\pi_k\epsilon(T)\}$\\
$\geq 1-\sum\limits_{k=1}^{K}\mathbb{P}\{N_k<(1-\delta)\pi_k\epsilon(T)\}$\\
$\geq1-Ke^{-2\delta^2\epsilon(T)}\geq1-Ke^{-2logT}=1-KT^{-2}.$\\
\end{flushleft}
\indent On the other hand, the lower bound $(1-\delta)\pi_k\epsilon(T)\geq J+\frac{16JlogT}{(\Delta^*)^2}$, then
\begin{equation}
C_{k,j}\geq\lfloor1+\frac{16logT}{(\Delta^*)^2}\rfloor\geq\frac{16logT}{(\Delta^*)^2},\forall j\in\mathcal{A}.
\end{equation}
\indent Therefore,
\begin{equation}
\mathbb{P}\{\forall k\in\mathcal{X}, \forall j\in\mathcal{A}, C_{k,j}\geq\frac{16logT}{(\Delta^*)^2}\}\geq1-JT^{-2}
\end{equation}
\indent Next, we study the relationship between the estimates $\bar{\varepsilon}_{k_1,j_{11},j_{12}}$ and $\bar{\varepsilon}_{k_2,j_{21},j_{22}}$ at the end of the exploration stage. We note that
\begin{flushleft}
$\bar{\varepsilon}_{k_1,j_{11},j_{12}}\geq\bar{\varepsilon}_{k_2,j_{21},j_{22}}$\\
$\Leftrightarrow(\bar{\varepsilon}_{k_1,j_{11},j_{12}}-\varepsilon_{k_1,j_{11},j_{12}}-\frac{\varepsilon_{k_2,j_{21},j_{22}}-\varepsilon_{k_1,j_{11},j_{12}}}{2})$\\
$-(\bar{\varepsilon}_{k_2,j_{21},j_{22}}-\varepsilon_{k_2,j_{21},j_{22}}-\frac{\varepsilon_{k_2,j_{21},j_{22}}-\varepsilon_{k_1,j_{11},j_{12}}}{2})\geq 0$\\
$\Leftrightarrow (\frac{\bar{u}_{k_1,j_{11}}-u_{k_1,j_{11}}}{c_{k_1,j_11}-c_{k_1,j_{12}}}-\frac{\varepsilon_{k_2,j_{21},j_{22}}-\varepsilon_{k_1,j_{11},j_{12}}}{4})$\\
$-(\frac{\bar{u}_{k_1,j_{12}}-u_{k_1,j_{12}}}{c_{k_1,j_{11}}-c_{k_1,j_{12}}}+\frac{\varepsilon_{k_2,j_{21},j_{22}}-\varepsilon_{k_1,j_{11},j_{12}}}{4})$\\
$- (\frac{\bar{u}_{k_2,j_{21}}-u_{k_2,j_{21}}}{c_{k_2,j_{21}}-c_{k_2,j_{22}}}+\frac{\varepsilon_{k_2,j_{21},j_{22}}-\varepsilon_{k_1,j_{11},j_{12}}}{4})$\\
$+(\frac{\bar{u}_{k_2,j_{22}}-u_{k_2,j_{22}}}{c_{k_2,j_21}-c_{k_1,j_{22}}}-\frac{\varepsilon_{k_2,j_{21},j_{22}}-\varepsilon_{k_1,j_{11},j_{12}}}{4})\geq 0$.
\end{flushleft}
\indent Thus, for the event $\bar{\varepsilon}_{k_1,j_{11},j_{12}}\geq\bar{\varepsilon}_{k_2,j_{21},j_{22}}$ to be true, we require that at least one term (with the sign) in the last inequality above is no less than zero. Conditioned on $C_{k,j}\geq\frac{16logT}{(\Delta^*)^2}$, we can bound the probability of each term according to the Hoeffding-Chefnoff bound, e.g., for the first term, we have
\begin{flushleft}
$\mathbb{P}\{\frac{\bar{u}_{k_1,j_{11}}-u_{k_1,j_{11}}}{c_{k_1,j_11}-c_{k_1,j_{12}}}-\frac{\varepsilon_{k_2,j_{21},j_{22}}-\varepsilon_{k_1,j_{11},j_{12}}}{4}\geq 0|C_{k_1,j_{11}\geq\frac{16logT}{(\Delta^*)^2}}\}$\\
$\leq\mathbb{P}\{\bar{u}_{k_1,j_{11}\geq u_{k_1,j_{11}}+\frac{\Delta^*}{4}|C_{k_1,j_{11}}}\geq \frac{16logT}{(\Delta^*)^2}\}$\\
$\leq e^{-2logT} =T^2$.
\end{flushleft}
\indent The conclusion then follows by considering the event $\{C_{k,j}\geq \frac{16logT}{(\Delta^*)^2},\forall k\in\mathcal{X},\forall j\in\mathcal{X}\}$ and its negation.
\end{proof}

\begin{theorem}\label{theorem 5}
Let $0<\delta<1$. Under $\epsilon$-ESs-JamNet-UCB-AILP, if
$$\epsilon(T)\geq\frac{J}{(1-\delta)\pi_{\textup{min}}}+logT\textup{max}\{\frac{1}{\delta^2},\frac{16J}{(1-\delta)\pi_{\textup{min}}(\Delta^*)^2}\},$$
then the regret of $\epsilon$-ESs-JamNet-UCB-AILP satisfies:
\begin{itemize}
\item if $\rho=\bar{E}/T\neq Q_i$, then $R_{\epsilon-FirstALP}(T,\bar{E})=O(logT)$;
\item if $\rho=\bar{E}/T= Q_i$,then $R_{\epsilon-FirstALP}(T,\bar{E})=O(\sqrt{T})$.
\end{itemize}
\end{theorem}
\begin{proof}
(Sketch) The key idea of proving this theorem is considering the event where the $\varepsilon_{k,j_1,j_2}$'s are ranked correctly and its negation. When the $\varepsilon_{k,j_1,j_2}$'s are ranked correctly, we can use the properties of the ALP algorithm with modification on the time horizon and energy budget (subtracting the time and energy budget in the exploration stage, which is $O(logT)$); otherwise, if the scheduler  obtains a wrong ranking results, the regret is bounded as $O(1)$ because the probability is $O(T^{-2})$ and the HE in each slot is bounded.
\end{proof}
\subsection{A Practical Implementation: Determine $\epsilon(T)$ without Prior Information}
\indent In Theorem  \ref{theorem 5}, the scheduler  requires the value of $\Delta^*$ (in fact $\Delta_{\textup{min}}^{(\varepsilon)}$ because  $\Delta_{\textup{min}}^{(c)}$ is known) to calculate $\epsilon(T)$. This is usually impractical since the expected energy rewards are unknown a \emph{priori}. Thus, without the knowledge of $\Delta_{\textup{min}}^{(\varepsilon)}$, we propose a Confidence Level Estimation (CLE) algorithm for deciding when to end the exploration stage.

Specifically, assume $\Delta_{\textup{min}}^{(\varepsilon)}>0$ and is unknown by the scheduler . In each slot of the exploration stage, the scheduler  tries to solve the problem ($\mathcal{LP}'_{\tau,b}$) with $u_{k,j}$ replaced by $\bar{u}_{k,j}$ using comparison, i.e., using Algorithm \ref{algorithm 4} and sorting the virtual jammers. For each comparison, the scheduler  tests the confidence level according th Algorithm \ref{algorithm 6}. If all comparisons pass the test, i.e., { \large{flagSucc = true} } for all comparisons, then the scheduler  ends the exploration stage and starts the exploitation stage.
\begin{algorithm}[h]
\caption{Confidence Level Estimation (CLE)}\label{algorithm 6}
\begin{algorithmic}[1]
\REQUIRE Time horizon $T$, estimate $\bar{\varepsilon}_{k_1,j_{11},j_{12}},\bar{\varepsilon}_{k_2,j_{21},j_{22}}$, number of executions $C_{k_1,j_{11}},C_{k_1, j_{12}},C_{k_2, j_{21}}$, and $C_{k_2, j_{22}}$;\
\ENSURE   {\large{flagSucc}};\
\quad $\Delta'=\frac{\Delta_{\textup{min}^{(c)}}(\bar{\varepsilon}_{k_1,j_{11},j_{12}}-\bar{\varepsilon}_{k_2,j_{21},j_{22}})}{2}$;\
\IF{$e^{-2(\Delta')^2\textup{min}\{C_{k_1,j_{11}},C_{k_1.j_{12}}\}}\leq T^{-2}\&e^{-2(\Delta')^2\textup{min}\{C_{k_2,j_{21}},C_{k_2.j_{22}}\}}\leq T^{-2}$}
\STATE {\large{flagSucc=true}};\
\ENDIF
\RETURN {\large{flagSucc}};\

\end{algorithmic}
\end{algorithm}

By similar arguments as in Theorem 4, we can show that the $\epsilon$-First policy with CLE will achieve $O(logT)$ regret except for the boundary cases, where it achieves $O(\sqrt{T})$ regret. On one hand, according to Hoeffding-Chernoff bound, if all comparisons pass the confidence level test, then with probability at least $1-KJ^2T^{-2}$, the algorithm obtains the correct rank and provide a right solution for the problem $(\mathcal{LP}'_{\tau,b})$. On the other hand, because $\Delta^*>0$, from the analysis in the previous section, we know that the exploration stage will end within $O(logT)$ rounds with high probability. Therefore, the expected regret is the same as that in the case with known $\Delta_{\textup{min}}^{(\epsilon)}$.

\section{Online Jammer Placement and Power Allocation: A Simple Algorithm and a Upper Bound}
In our online jammer placement model, we receive an unknown number of $J$ \emph{friendly jammer
placement requests} sequentially over time. Each jammer $1 \le j \le J$ targets at an eavesdropping location $p_e$ at $\mathfrak{F}$. Let
 the short notation $d_{jp_e(j)}$ denote the distance between the jammer $j$ and its jammed eavesdropping location  $p_e(j)$. We denote
 $\Delta = (\max_j d_{jp_e(j)})/(\min_j d_{jp_e(j)})$ as the \emph{distance ratio}. Further, if jammers are informed with a parameter
 $t_j$, which denotes the eavesdropping duration within the jamming scope. We $\Gamma = (\max_i t_i)/(\min_i t_i)$ as the
  \emph{duration ratio}. W.l.o.g., we let $\min_i t_i =1$ and $\max_i t_i= \Gamma$.   Jammers arrive sequentially over time and
  the goal is to accept the minimal number of requests to interfering all potential
  eavesdropper while making minimal noises to legitimate communications.

  For each jammer
  placement request, an online algorithm must decide whether to accept the request or deny it, and the decision
  can not be revoked. For an accepted jammer $j$ it needs to set a power level $P_j$ and a
  channel $f_j \in \{1,...,F\}$ to it to emit the interfering power.  In the following  we first analyze the spatial aspect of the problem and assume that eavesdropping behavior last forever, i.e., $t_j = \infty$. We begin by analyzing a simple online algorithm for the case of a single channel and any polynomial power assignment. Our analysis of the online algorithm introduces a number of critical observations that are used in later subsections.

\begin{figure}
\centering
\includegraphics[scale=.5]{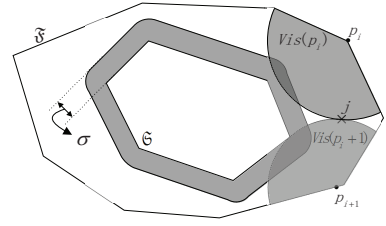}
\caption{Guarded safe distance $\sigma$ for online jammer placement.}
\label{fig:digraph}
\vspace{-.2cm}
\end{figure}


 The main idea of the algorithm is to accept a new jammer
  only if it keeps a $safe\  distance\  \sigma$ from every other previously accepted jammers to
  meet two goals: 1) the sum of cumulated interfering power of all jammers to any legitimate
  communications are small enough; 2) the maximized safe distance $\sigma$ make the number of
  jammers placed to be minimized. In particular, we accept incoming jammer $i$ only if
  $\min\{d_{j,p_s}, \forall p_s \in \partial\mathfrak{S}, \forall j \in J\} \geq \sigma$ and $\max\{d_{i,p_e(j)},d_{j,p_e(i)}\} \geq \sigma$ for every other previously accepted jammer $j\in J$.  We call this algorithm J{\footnotesize AM}-S{\footnotesize AFE}-D{\footnotesize ISTANCE}. There is a important
   tradeoff for the choice of $\sigma$ among EE of interfering power,  validity and competitive ratio. A larger $\sigma$
   means safeguarding a large scope of eavesdropping locations  with high interfering power and have resulted in minimized number of
   jammers to place, but the SIR
   constraints of legitimate communications might become violated. If $\sigma$ is too small, then more jammers
   are required to be placed and at some point the accumulated interference at an accepted jammer placement
   request can get too large and the SIR constraint of legitimate communications  becomes violated.

 We strive to devise the J{\footnotesize AM}-S{\footnotesize AFE}-D{\footnotesize ISTANCE} to
  make $\sigma$ as large as possible to ensure optimal competitive ratio as well as not to violate
   SIR constraints. On the one hand, we need to bound the interference on the
   edge of the fence $\mathfrak{S}$  at accepted jammer placement requests to construct a worst-case legitimate communication
   scenario. On the other hand,  we consider  an  accepted jammer $j$  block an eavesdropper $p_e(j)$ with certain distance $r_j$. In the following we show that for $r \in [0,1]$ the choice of
   \begin{equation}
   \sigma\! =\! \min\left\{\!2\Delta,  \max \left\{\!4  \Delta^r  \sqrt[\gamma]{\frac{ 72\delta_s}{\bar P (\gamma-2)}}, \Delta^{(1-r)} \sqrt[\gamma] {\frac{\tilde P}{\delta_e}} \right\} \!\right\}
   \label{eq:BD0}
   \end{equation}
      is sufficient to yield the Theorem 5 in the following. Denote the $\underline L = \min ||\mathfrak{S} - \mathfrak{F} ||$ as the
      minimal distance among $\mathfrak{S}$ and $\mathfrak{F}$. If the eavesdropping locations on the $\mathfrak{F}$ are unavailable,
      we need the additional condition that
     \begin{equation}
\quad\quad\quad\quad\quad\quad\quad\quad \underline L \ge (\sqrt 2 +2)\sigma
   \label{eq:BD90}
   \end{equation}
   to  block all eavesdropping points on $\mathfrak{F}$.

\begin{theorem}\label{theorem1}
In a single channel scenario, J{\footnotesize AM}-S{\footnotesize AFE}-D{\footnotesize ISTANCE} is $\Omega(\Delta)$-competitive for any polynomial
 interfering power assignment with $r\in[0,1]$.
\end{theorem}
\begin{proof}
We first show that J{\footnotesize AM}-S{\footnotesize AFE}-D{\footnotesize ISTANCE} is valid, i.e., for an accepted jammer $j$ the SIR constraint of any
locations at the edge of $\mathfrak{S}$ never becomes violated. In particular, we will underestimate the distances of already accepted jammers to overestimate the interference at any position $p_s, \forall p_s \in \partial\mathfrak{S}$. As such, even under the wrest conditions the SIR
 constraint at any potential legitimate receiver, i.e.,  $p_s$,  will remain valid.

 To estimate the interference at $p_s$, we have to calculate how many jammers may be placed at which distance.
 Using the fact $\underline L \geq 2\Delta \geq \sigma $ as shown in  Fig. 3.
 Then, it is straightforward to devise the  rule of the algorithm such  that any two different accepted jammers are at least a distance of $\sigma - \Delta \leq \sigma/2$ apart to block eavesdroppers on the $\mathfrak{F}$. We segment all of $\mathbb{R}^2$ into \emph{2}-dimensional squares with length $\sigma/4$ and we call it \emph{sectors}. The greatest distance within a sector is $\sigma \sqrt 2/4 = \sigma/2\sqrt 2 \leq \sigma/2$. Each sector can contain jammer  from at most one request, so there are at most two jammers in every sector.

 W.l.o.g., we assume that sectors are created such that the jammer $j$ lies in a corner point of $2^2$ sectors. We divide the set of sectors into $layers$. The first layer consists of the $2^2$ sectors incident to $j$. The second layer are all sectors not within the first layer but share at least a point with sector from the first layer, and so on. Hence, there are  $(2l)^2$ sectors from layers 1 through $l$, and their union is a large square of side length $2\emph{l}\sigma/4$ with $j$ in the center. Therefore, there are exactly $2^2(l^2-(l-1)^2)$ sectors in layer $l$.
\begin{figure}
\centering
\includegraphics[scale=.5]{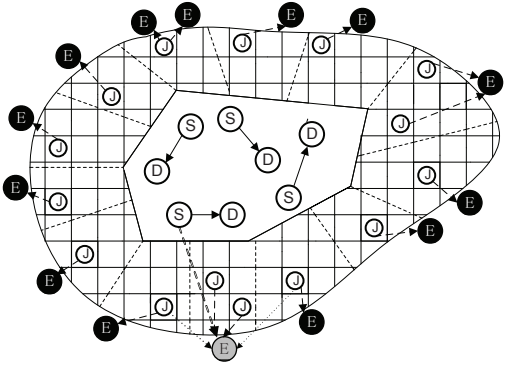}
\caption{Interference contributed from each segmented Sectors.}
\label{fig:digraph}
\vspace{-.2cm}
\end{figure}
Due to the algorithm there can be no sender at a distance smaller than $\sigma$ from $j$. The sector of smallest layer that is at a distance at least $\sigma$ from $j$ can be reached along the diagonal of the squares of that layer. There can be no jammer in all sectors from layers 1 through $l'$, where $l'$ is bounded by $\sigma \leq l'(\sigma/2\sqrt2)$, which yields $l' \geq 3$. For bounding the interference assume that in all sectors of layer $l\geq 3$ there are two jammers. Note that all jammer in sectors from a layer $l$ have a distance at least $(l-1)\sigma/4$ to $j$. To bound the interference that is created at $j$, we use the following technical lemma from \cite{Online13} under $\mathbb{R}^2$.
\end{proof}

\begin{lemma}  For $\gamma >2 \geq 1$ the following holds:$$2^2 \cdot \sum^{\infty}_{l=3}\frac{l^2-(l-1)^2}{(l-1)^\gamma}<\frac{36}{\gamma -2}.$$
\end{lemma}

With Lemma 6 and set $P_j=d_{jp_e(j)}^{r\gamma}$, we bound the interference for
legitimate communication
\begin{equation*}
\begin{aligned} 
I\!\!=\!\! \sum_{j\in J}\!\! \frac{d_{jp_e(j)}^{r\gamma}}{d_{jp_s}^{\gamma}}\!\!< \!\! 2\Delta^{r\gamma}\sum_{l=3}^{\infty}\frac{2^2(l^2\!\!-\!\!(l-1)^2)}{((l-1)\sigma/4)^{\sigma}} \!\!<\!\!2\Delta^{r\gamma}(\frac{4}{\sigma})^{\gamma}\cdot \frac{36}{\gamma -2}.
\end{aligned}
\end{equation*}
To satisfy the SIR constraint at $p_s$, we let  $\bar P \geq \delta_s I$, i.e.,
$$2\delta_s \Delta^{r\gamma}\cdot(\frac{4}{\sigma})^{\gamma}\cdot \frac{36}{\gamma -2}\leq \bar P .$$
This yields a lower bound for the distance $\sigma$,
\begin{equation}
\quad\quad\quad\quad\quad\quad \sigma \geq  4 \cdot \Delta^r \cdot \sqrt[\gamma]{\frac{ 72\delta_s}{\bar P (\gamma-2)}},
\label{eq:BD1}
\end{equation}
which can be verified to hold for our choice of $\sigma$.

Then, to bound the SIR constraint at $p_e$, we use the fact that in the worst condition if there is no other jammer, a single jammer $j$ is enough to thwart the eavesdropper, i.e., $ d_{jp_e(j)}^{r\gamma} / d_{jp_e(j)}^{\gamma} \delta_e \geq {\tilde P/{  d_{s({p_e}){p_e}}^{  \gamma }}}.$ Note that
 $d_{jp_e(j)}^{r\gamma} / d_{jp_e(j)}^{\gamma} = d_{jp_e(j)}^{(r-1)\gamma} \ge \Delta^{(r-1)\gamma}$ and $  d_{s({p_e}){p_e}} \ge \underline L \ge \sigma$, we have
\begin{equation}
\quad\quad\quad\quad\quad\quad
{\Delta^{(r-1)\gamma} } \delta_e  \ge {\frac{\tilde P}{  \sigma^{\gamma }}}.
\IEEEyesnumber \label{eq:BD23}
\end{equation}
This yields another lower bound for $\sigma$
\begin{equation}
\quad\quad\quad\quad\quad\quad
\sigma \ge \Delta^{(1-r)} \sqrt[\gamma] {\frac{\tilde P}{\delta_e}}.
\IEEEyesnumber \label{eq:BD2}
\end{equation}
Combine results (\ref{eq:BD1}) and (\ref{eq:BD2}) yields  (\ref{eq:BD0}).
%

Moreover, when the locations of eavesdroppers are unavailable, every placed jammer is
necessary  to block all eavesdropping position on the intersections of its sector and the fence $\mathfrak{S}$.
Denote the furthest eavesdropping position within the sector as $p_e'$. Then we have
$d_{jp_e(j)}^{r\gamma} / d_{jp_e'}^{\gamma} \delta_e \ge d_{jp_e(j)}^{r\gamma} / (d_{jp_e}+d_{p_e'p_e})^{\gamma} \delta_e$. Due to
the size of the sector we have that $d_{p_e'p_e} \le \sqrt 2$. Also $d_{jp_e} \ge 1$, which
implies
\begin{equation}
\frac{d_{jp_e(j)}^{r\gamma}\delta_e}{ (d_{jp_e}+d_{p_e'p_e})^{\gamma}}\ge \frac{1}{(\sqrt 2+1)^\gamma} \frac{d_{jp_e(j)}^{r\gamma} }{ d_{jp_e}^{\gamma}}\delta_e  \ge
 \frac{ {\Delta^{(r-1)\gamma} } \delta_e}{(\sqrt 2+1)^\gamma}.
\IEEEyesnumber \label{eq:BD2}
\end{equation}
  On the other hand, use the condition   (\ref{eq:BD90}) that we upper bound
  ${\tilde P/{  d_{s({p_e'}){p_e'}}^{  \gamma }}} \le  {\tilde P/{  (\underline L -\sigma  )^{  \gamma }}}
  \le  \frac{\tilde P}{{(\sqrt 2+1)^\gamma}{  \sigma^{  \gamma }}}$. Use the fact that (\ref{eq:BD23}), we have
\begin{equation}
\quad\quad\quad\quad\quad  \frac{ {\Delta^{(r-1)\gamma} } \delta_e}{(\sqrt 2+1)^\gamma} \ge  \frac{\tilde P}{{(\sqrt 2+1)^\gamma}{  \sigma^{  \gamma }}}.
\IEEEyesnumber \label{eq:BD21}
\end{equation}

To bound the competitive ratio we need the following \emph{Density Lemma}, which is motivated by
 Lemma 3 in Andrews and Dinitz\cite{2009} to restrict interference both from senders and at receivers for
 any legitimate communication links. However, the placement of friendly jammer is
 interesting to be found as a different problem. In this case, we need to estimate the
 interference caused at the legitimate communication of a placed jammer by mapping its transmission power
 calculated from the SIR constraint at the eavesdropper. 

\begin{lemma}
(Density Lemma) Assume a sector $A$ with side-length $x\geq 1$ and any feasible jammer placement solution with arbitrary power assignment. There can be only $\frac{{{\delta _e}{3^\gamma }{{\bar L}^\lambda }}}{{{\delta _s}}}\frac{{\bar P}}{{\tilde P}}(x+1)^2$ jammer placement requests in $A$.
\end{lemma}
\begin{proof}
We first assume $x=1$ and consider the number of jammers in section $A$. At first, the interfering power
 receiving by from a jammer $j$ to its targeted eavesdropper $p_e(j)$ is a constant $\bar p = {P_j}/{d_{jp_e(j)}^\gamma}$ such that $\bar p \delta_e \ge
  {\tilde P/{  d_{s({p_e}){p_e}}^{  \gamma }}}$  for any jammer placement request $j$ within $A$. Consider the
  interfering power contributed to the legitimate communication by the same jammer $j$. Due to the fact that  $\max d_{j{p_e}p_e(j)}\le 2 \min d_{js(p_e)}$. Also $d_{jp_e(j)}\geq 1$, which implies
  \begin{equation}
  \begin{array}{l}
  \frac{{{P_j}}}{{d_{js({p_e})}^\gamma }} \ge \frac{{{P_j}}}{{{{\left( {d_{j{p_e}s({p_e})}^{} + d_{j{p_e}(j)}^{}} \right)}^\gamma }}}
  \ge \frac{1}{{{{(1 + 2)}^\gamma }}}\bar p \\
   \quad\quad\quad \ge \frac{1}{{{\delta _e}{3^\gamma }}}\frac{{\tilde P}}{{d_{s({p_e}){p_e}(j)}^\gamma }} \ge \frac{1}{{{\delta _e}{3^\gamma }}}\frac{{\tilde P}}{{{{\bar L}^\lambda }}},
  \end{array}
  \end{equation}
where we have the $\bar L = \max ||\mathfrak{S} - \mathfrak{F} ||$.  Thus, if more than $\frac{{{\delta _e}{3^\gamma }{{\bar L}^\lambda }}}{{{\delta _s}}}\frac{{\bar P}}{{\tilde P}}$ such placement request are present, the SIR constraint for the legitimate communication is violated. Now consider the arbitrary power allocation strategy. If we artificially reduce powers such that all the jamming links to $\mathfrak{F}$ experience a minimal signal strength $\bar{p}$, and then increase powers to their original value. The increase lowers SIR for the jammers that continue to have a signal strength of $\bar{p}$. Hence, if more than $\frac{{{\delta _e}{3^\gamma }{{\bar L}^\lambda }}}{{{\delta _s}}}\frac{{\bar P}}{{\tilde P}}$ jammers are present in $A$, \emph{at least one} SINR constraint is violated.
\end{proof}

 The density lemma allows a simple way to bound the number of jammers the optimum solution can accept in the blocked area. First consider a jammer $j$ of a placement request accepted by  J{\footnotesize AM}-S{\footnotesize AFE}-D{\footnotesize ISTANCE}. The jammer blocks a square of radius $\sigma$
 for eavesdroppers and limit the placement of other requests to reduce the interference to legitimate communications. We overestimate its size by a sector of side-length $2\sigma$ centered at $j$. By the density lemma, the optimum solution can accept at most $\frac{{{\delta _e}{3^\gamma }{{\bar L}^\lambda }}}{{{\delta _s}}}\frac{{\bar P}}{{\tilde P}}(2\sigma+1)^2$ jammers, which is $\Omega(\Delta)$ according to    (\ref{eq:BD0}) ($\sigma=\Omega(\Delta^{1/2})$) for fixed $\delta_s, \delta_e$ and $\gamma$. Finally, note that $\sigma$ is chosen to maximize the competitive ratio and does not optimize its involved constants.

  In the next, we use similar arguments to show a result for any other polynomial power assignment. As safe distance we pick $\sigma^+ = \Delta^r \cdot \sigma$ if $r>1$, and $\sigma^- = \Delta^{1-r}\cdot \sigma$ if $r < 0$.

\begin{figure}
\centering
\includegraphics[scale=.5]{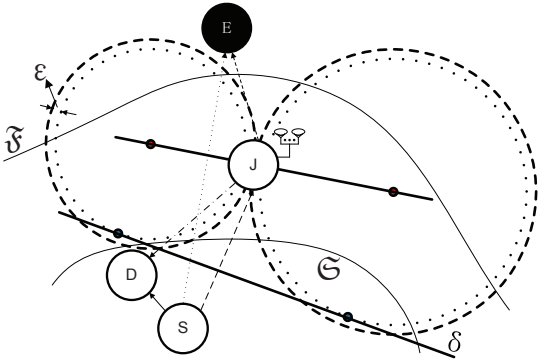}
\caption{$\epsilon$-approximation of power assignment without priori of exact eavesdropper locations.  }
\label{fig:digraph}
\vspace{-.2cm}
\end{figure}

\begin{corollary}
 J{\footnotesize AM}-S{\footnotesize AFE}-D{\footnotesize ISTANCE} is $\Omega(\Delta^{2\cdot max\{r,1-r\}})$-competitive for a polynomial power assignment with $r\notin(0,1)$ and a single channel.
\end{corollary}
\begin{proof}
In the case $r>1$ we note for validity of the algorithm that the interference at an accepted jammer $j$ is again bounded by
\begin{equation*}
\begin{aligned} 
I\!\!=\!\! \sum_{j\in J}\!\! \frac{d_{jp_e(j)}^{r\gamma}}{d_{jp_s}^{\gamma}}\leq\Delta^{r\gamma} \sum_{j\in J} \frac{1}{d_{jp_s}^{\gamma}} <&2\Delta^{r\gamma}\cdot(\frac{4}{\sigma^+})^{\gamma}\cdot\frac{36}{(\gamma - 2)}.
\end{aligned}
\end{equation*}
\indent The SIR constraint now requires that   $\bar P \geq \delta_s I$ at $p_s$. This yields a lower bound of
\begin{equation}
\quad\quad\quad\quad
\sigma^+\geq 4 \cdot \Delta^r \cdot \sqrt[\gamma]{\frac{ 72\delta_s}{\bar P (\gamma-2)}}.
\IEEEyesnumber \label{eq:BD122}
\end{equation}
\indent If $r<0$, then the interference is maximized with requests of length 1 in each sector. The interference is thus bounded by
$$I\!\!=\!\! \sum_{j\in J}\!\! \frac{d_{jp_e(j)}^{r\gamma}}{d_{jp_s}^{\gamma}}\leq\Delta^{r\gamma} \sum_{j\in J} \frac{1}{d_{jp_s}^{\gamma}} <2\cdot(\frac{4}{\sigma^-})^{\gamma}\cdot \frac{36}{\gamma-d}.$$

\indent The SIR constraint now requires that   $\bar P \geq \delta_s I$ at $p_s$. This yields a lower bound
\begin{equation}
\quad\quad\quad\quad\quad\quad
\sigma^-\geq   4 \cdot \sqrt[\gamma]{\frac{ 72\delta_s}{\bar P (\gamma-2)}}.
\IEEEyesnumber \label{eq:BD222}
\end{equation}

Similarly, we need to bound the SIR constraint at $p_e$, we use the fact that in the worst condition if there is no other jammer, a single jammer $j$ is enough to thwart the eavesdropper. In the case $r>1$, note that
 $d_{jp_e(j)}^{r\gamma} / d_{jp_e(j)}^{\gamma} \delta_e  = d_{jp_e(j)}^{(r-1)\gamma} \delta_e  \ge \Delta^{(r-1)\gamma} \delta_e  \ge \delta_e  \geq {\tilde P/{  d_{s({p_e}){p_e}}^{  \gamma }}}$ and $  d_{s({p_e}){p_e}} \ge \underline L \ge \sigma$, we have
$\delta_e  \ge {\frac{\tilde P}{  {(\sigma^+)}^{\gamma }}}.
$
This yields another lower bound for $\sigma^+$
\begin{equation}
\quad\quad\quad\quad\quad\quad\quad\quad
\sigma^+ \ge \sqrt[\gamma] {\frac{\tilde P}{\delta_e}}.
\IEEEyesnumber \label{eq:BD322}
\end{equation}
If $r <0$, we have
 $d_{jp_e(j)}^{r\gamma} / d_{jp_e(j)}^{\gamma} \delta_e  = d_{jp_e(j)}^{(r-1)\gamma} \delta_e  \ge \Delta^{(r-1)\gamma} \delta_e     \geq {\tilde P/{  d_{s({p_e}){p_e}}^{  \gamma }}}$.
This yields another lower bound for $\sigma^-$
\begin{equation}
\quad\quad\quad\quad\quad\quad
\sigma^- \ge \Delta^{(1-r)} \sqrt[\gamma] {\frac{\tilde P}{\delta_e}}.
\IEEEyesnumber \label{eq:BD422}
\end{equation}
Combine results  (\ref{eq:BD122}),  (\ref{eq:BD222}) ,  (\ref{eq:BD322}) and (\ref{eq:BD422}), the corollary follows.
\end{proof}

Comparing the results in  Theorem 5 and Corollary 8, it shows that the competitive ratio of  J{\footnotesize AM}-S{\footnotesize AFE}-D{\footnotesize ISTANCE} is asymptotically optimal  for polynomial power assignments with $r\in(0,1)$. This includes both the uniform and linear power assignment. Next, we bound the competitive ratio for any deterministic online jammer placement algorithm using polynomial power assignments. This can be generalized to a lower bound for any distance-based power assignment.
\begin{theorem}
Every deterministic online jammer placement algorithm using polynomial power assignments has a competitive ratio (1) $O(\Delta)$-competitive  using distance-based power assignments.
\end{theorem}
\begin{proof}
The main ingredient in the proof is that every deterministic online jammer placement algorithm has to accept the first jammer that arrives, otherwise it risks having an unbounded competitive ratio. Note that the jamemrs can be replaced over time, we can repeat the following instance sufficiently often and keep a sufficiently large distance between the instances. In this way we can neglect the constant $a$ from the competitive ratio.\\
\indent Since all jammer placements are directed (to eavesdroppers) and use polynomial power assignment. Let the first request have length $\Delta$. From the SIR constraints,  we bound the minimum distance every other successful jammer request has to keep to the fence $\mathfrak{F}$.  This yields a blocked area in which the online jammer placement  algorithm is not able to accept any  request. We then count the maximum number of requests that can be placed into this 
region, where the optimum solution can accept simultaneously.

 To extend the previous arguments to arbitrary distance-based power assignments, we observe that the previous lower bound uses only requests of length 1 and $\Delta$. Let $\phi$ be the function of the distance-based power assignment, then $\phi(\Delta)$ is the power of the first request. The lower bound for this power assignment behaves exactly as for a polynomial assignment with $r = (log\phi(\Delta))/(\alpha log\Delta)$. \\
\indent Note that when a power assignment is not distance-based, it might assign different powers to small requests based on whether they are near the sender or the receiver of the first request. This is not helpful since the jammer have a direct interfering link to the eavesdroppers. In this case,  we create the same instance using only undirected requests. Then we get a blocked area of at least $\Omega(\Delta)$ for any polynomial power assignment around both points of the first request. Using the normalization of powers as before we observe that there is a blocked area of size $\Omega(\Delta)$ for any small request, \emph{no matter which power we assign to it}. This proves the theorem.
\end{proof}

\begin{figure*}\label{figure 5}
\begin{minipage}[ht]{0.3\linewidth}
\centering
 \subfigure[Storage/fence with candidate locations (small dots) and solution of Jamming-LP(Large dots)]{
 \label{fig:subfig:c}
\includegraphics[width=6cm]{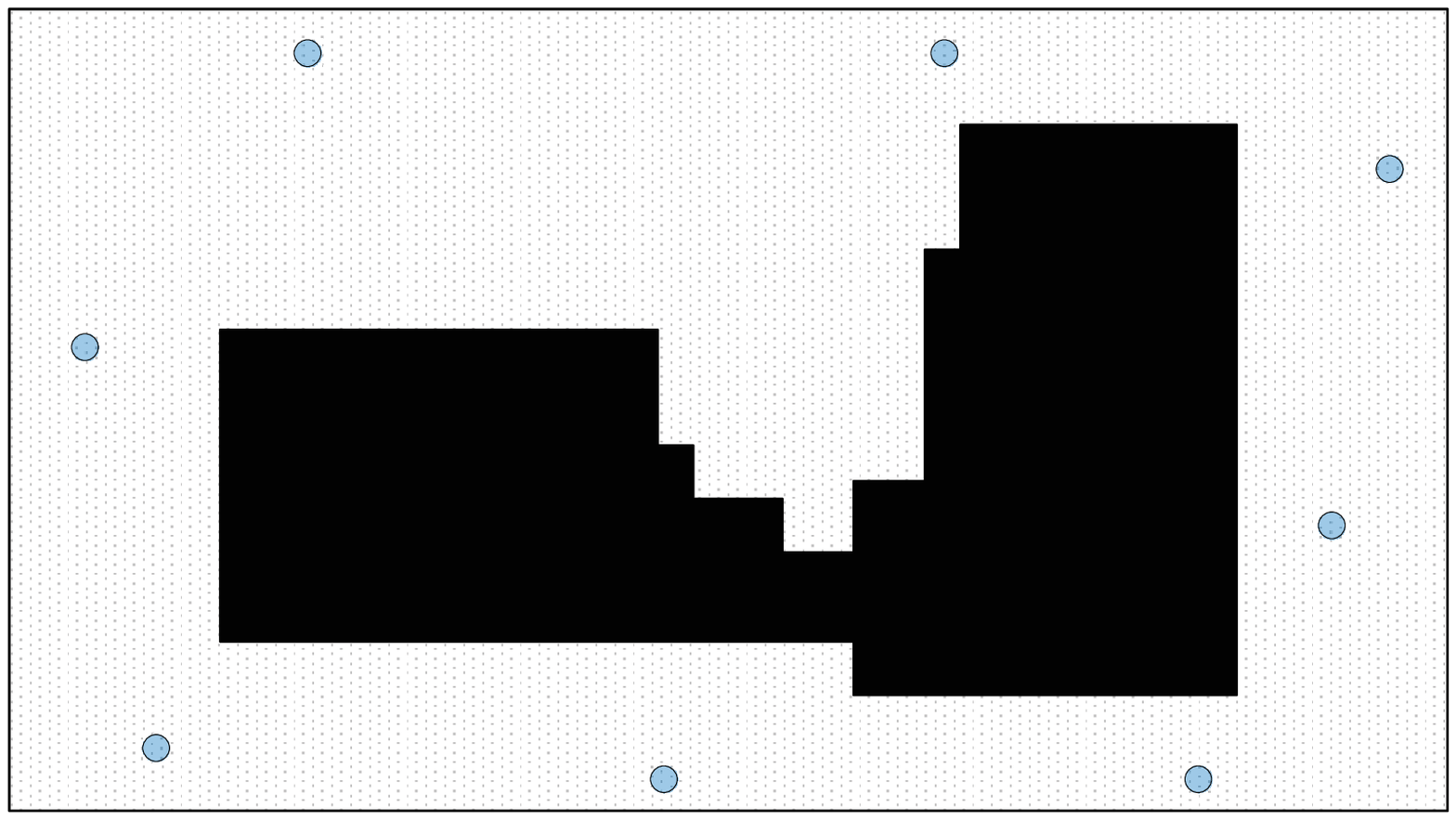}}
\end{minipage}%
\hfill
\begin{minipage}[ht]{0.35\linewidth}
\centering
 \subfigure[J{\footnotesize AM}-S{\footnotesize AFE}D{\footnotesize IST}-P{\footnotesize OWER}.]{
 \label{fig:subfig:a}
\includegraphics[width=5.5cm]{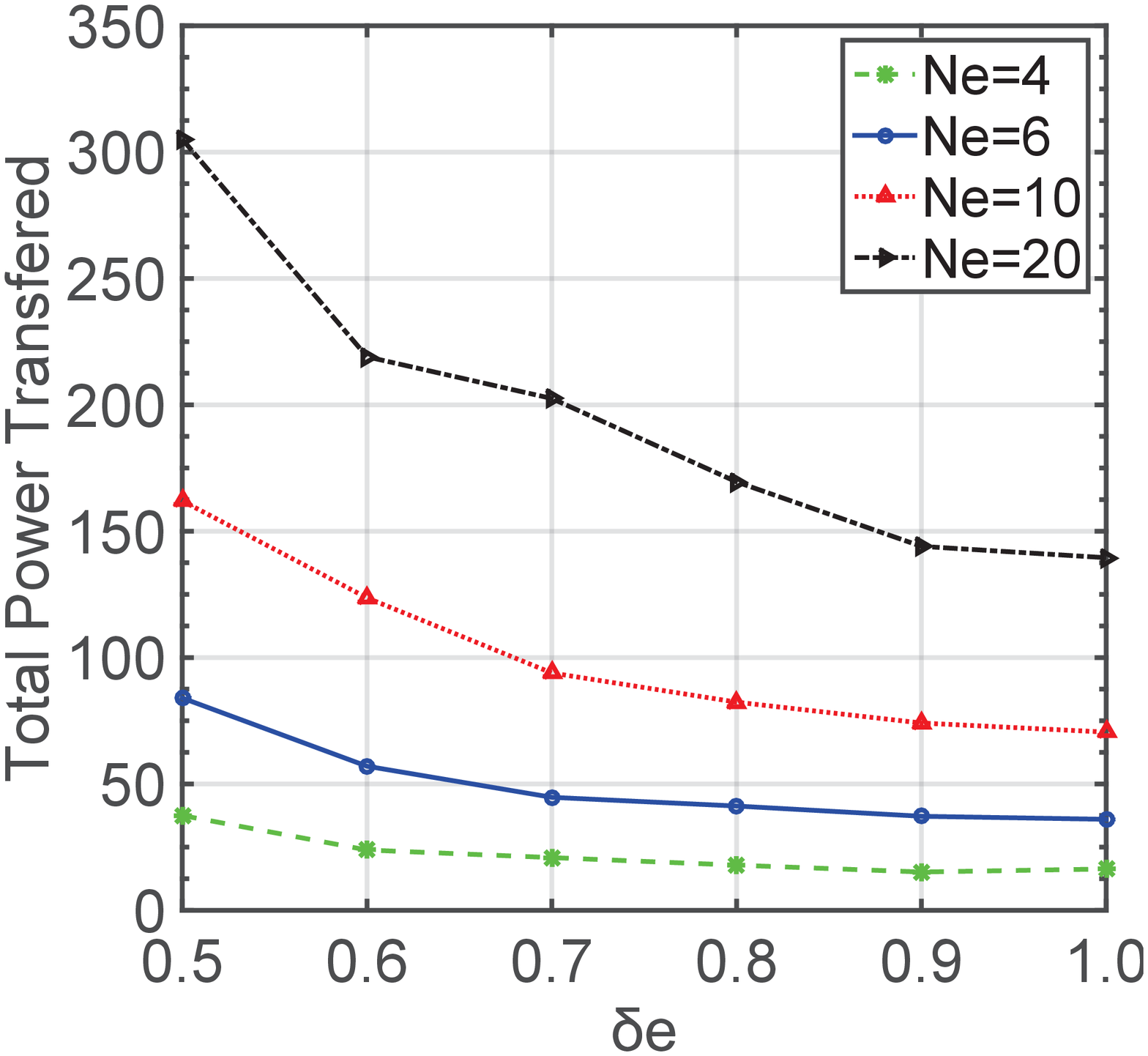}}
\end{minipage}
\hfill
\hspace{-.9cm}
\begin{minipage}[ht]{0.3\linewidth}
\centering
 \subfigure[J{\footnotesize AM}-L{\footnotesize IFE}M{\footnotesize AX}.]{
 \label{fig:subfig:b}
\includegraphics[width=5.6cm]{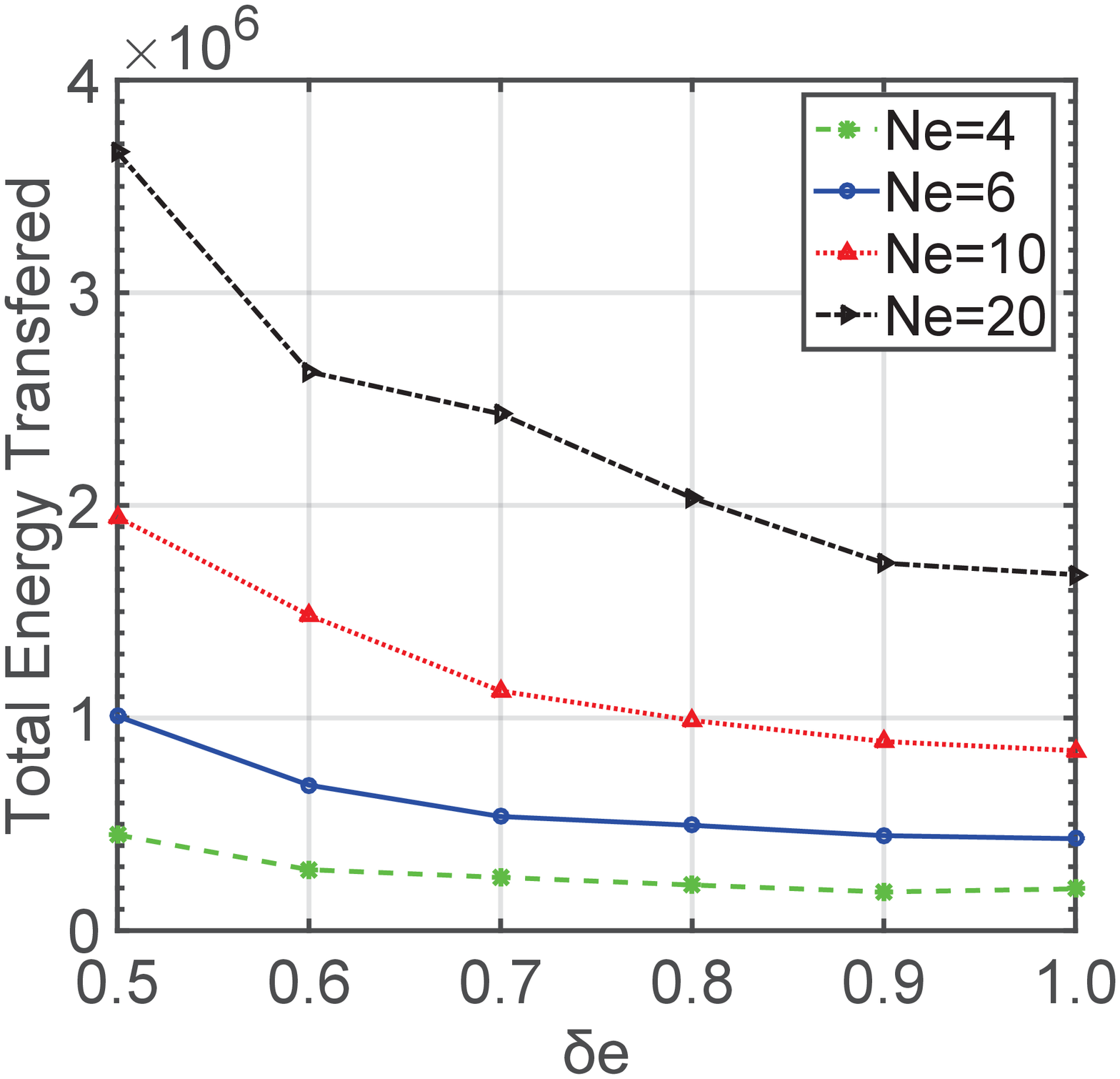}}
\end{minipage}
\hfill
\hspace{-.1cm}

\begin{minipage}[ht]{0.24\linewidth}
\centering
 \subfigure[J{\footnotesize AM}-S{\footnotesize AFE}-DistPower.]{
 \label{fig:subfig:d}
\includegraphics[width=4.43cm]{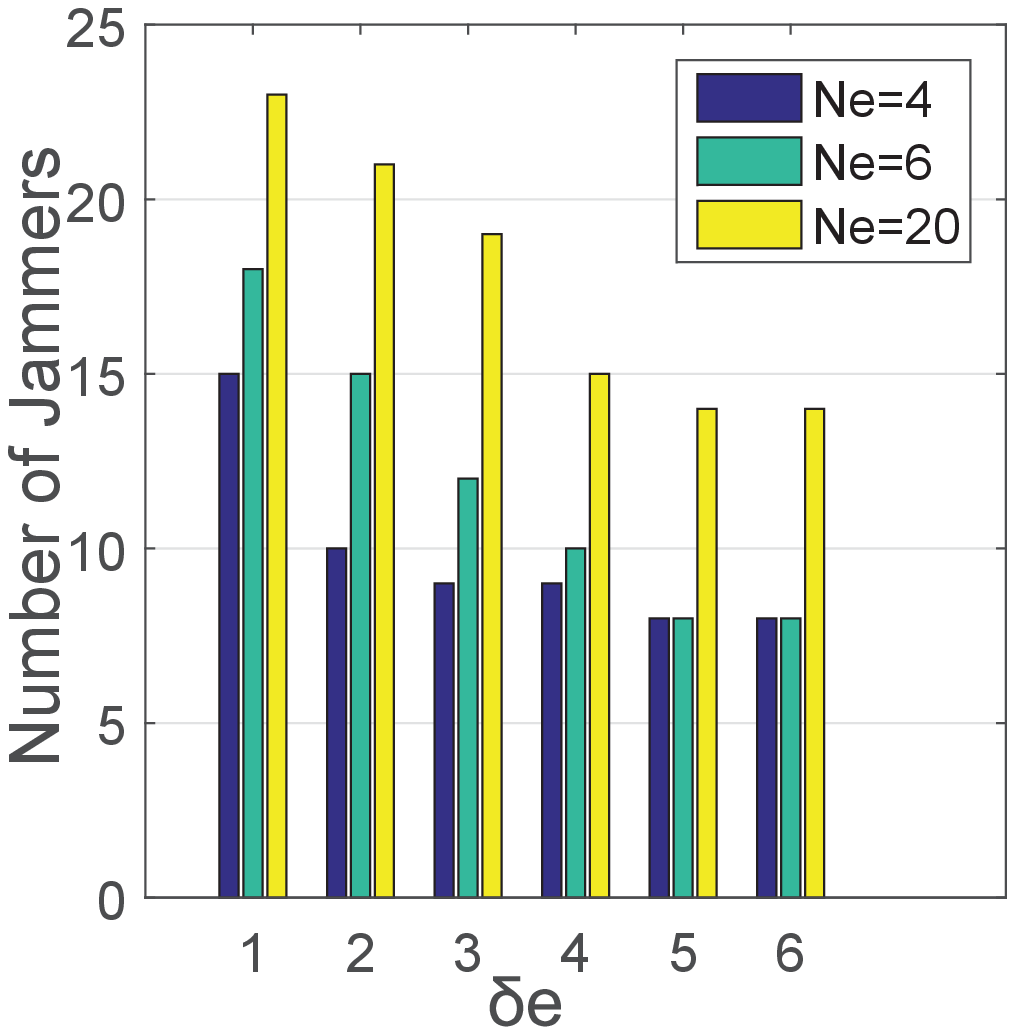}}
\end{minipage}
\hfill
\begin{minipage}[ht]{0.23\linewidth}
\centering
 \subfigure[J{\footnotesize AM}-L{\footnotesize IFE}-Max.]{
 \label{fig:subfig:e}
\includegraphics[width=4.6cm]{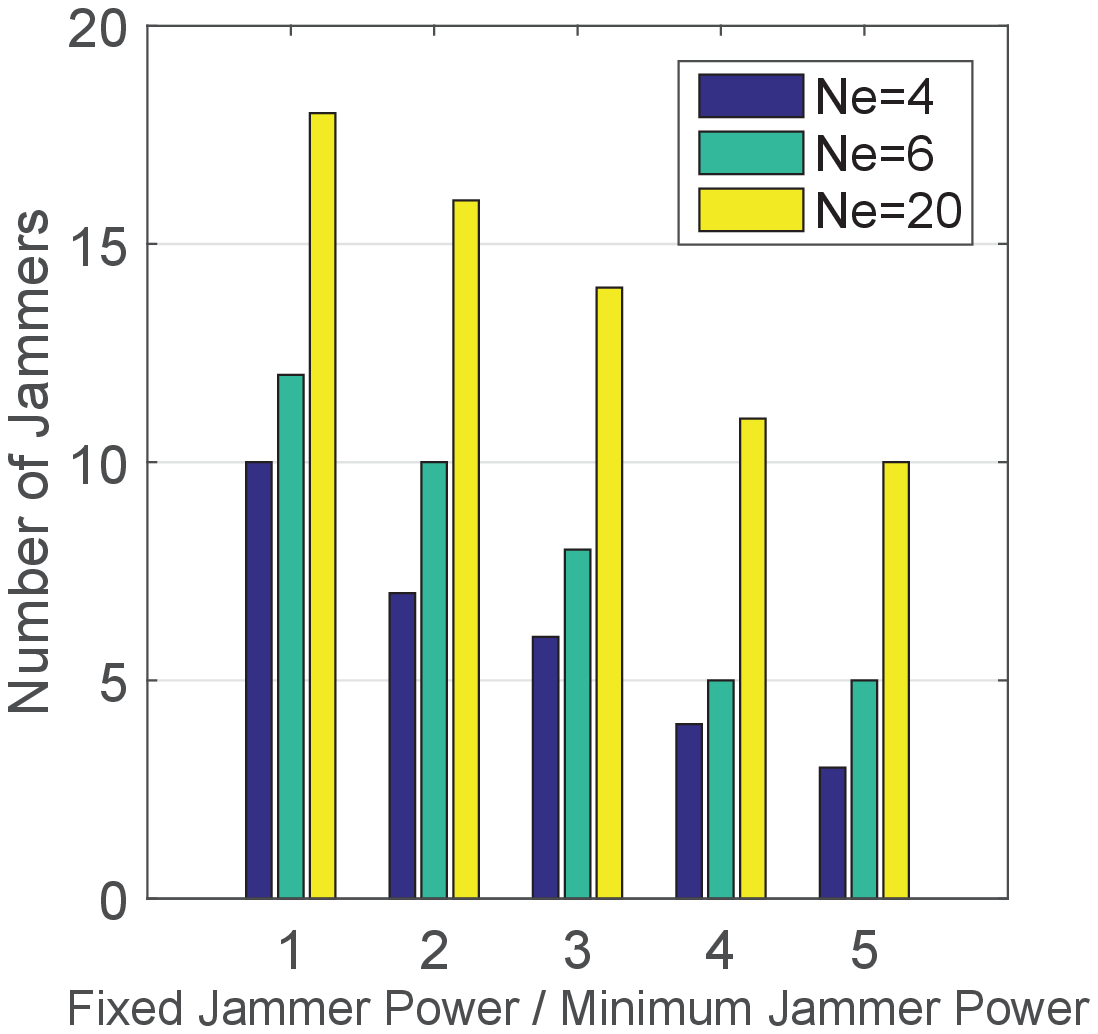}}
\end{minipage}
\hfill
\begin{minipage}[ht]{0.25\linewidth}
\centering
 \subfigure[J{\footnotesize AM}-S{\footnotesize AFE}D{\footnotesize IST}-MultiChannel.]{
 \label{fig:subfig:f}
\includegraphics[width=4.8cm]{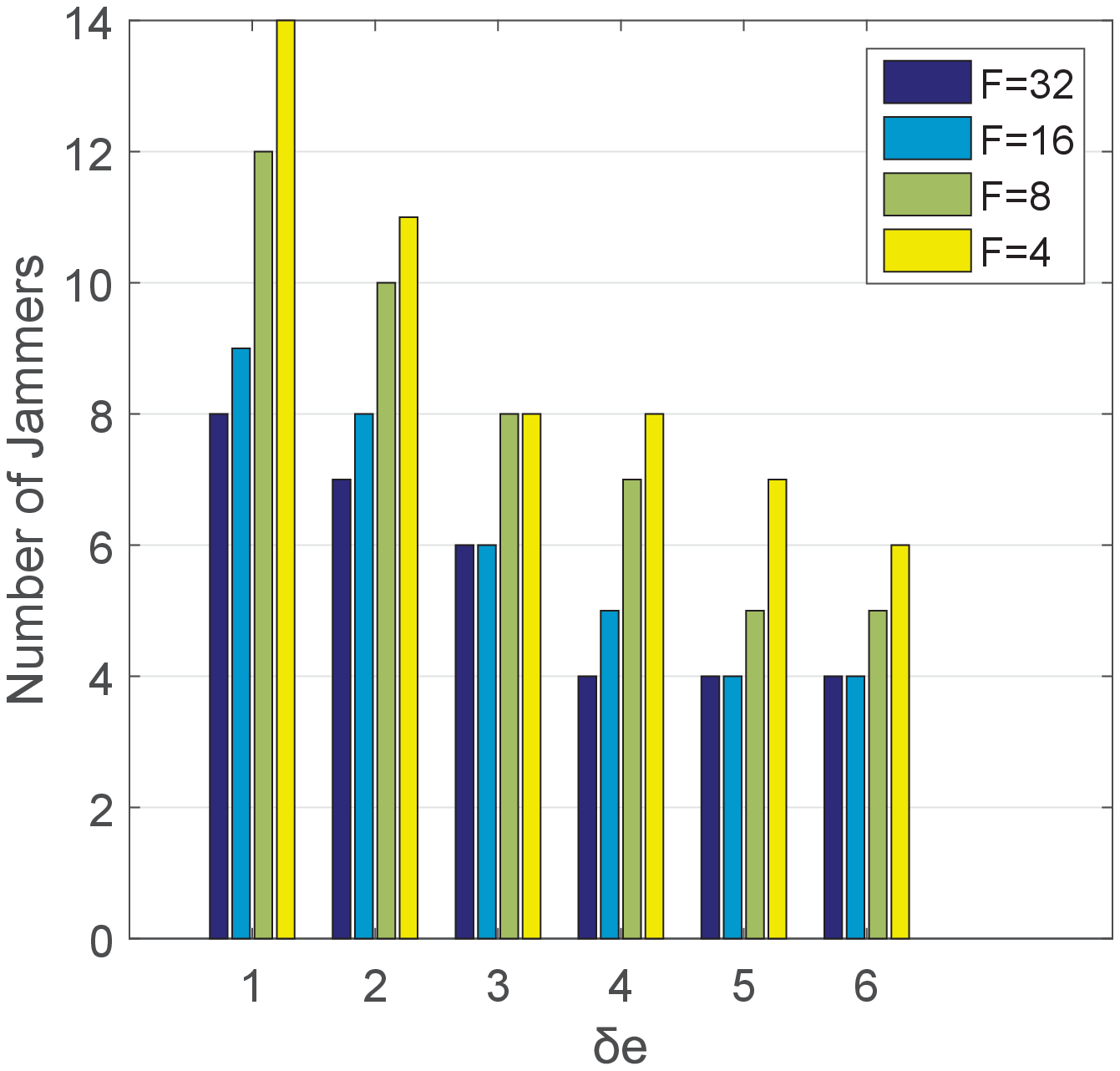}}
\end{minipage}%
\hfill
\begin{minipage}[ht]{0.25\linewidth}
\centering
 \subfigure[J{\footnotesize AM}-S{\footnotesize AFE}D{\footnotesize IST}-MultiChannel.]{
 \label{fig:subfig:f}
\includegraphics[width=4.53cm]{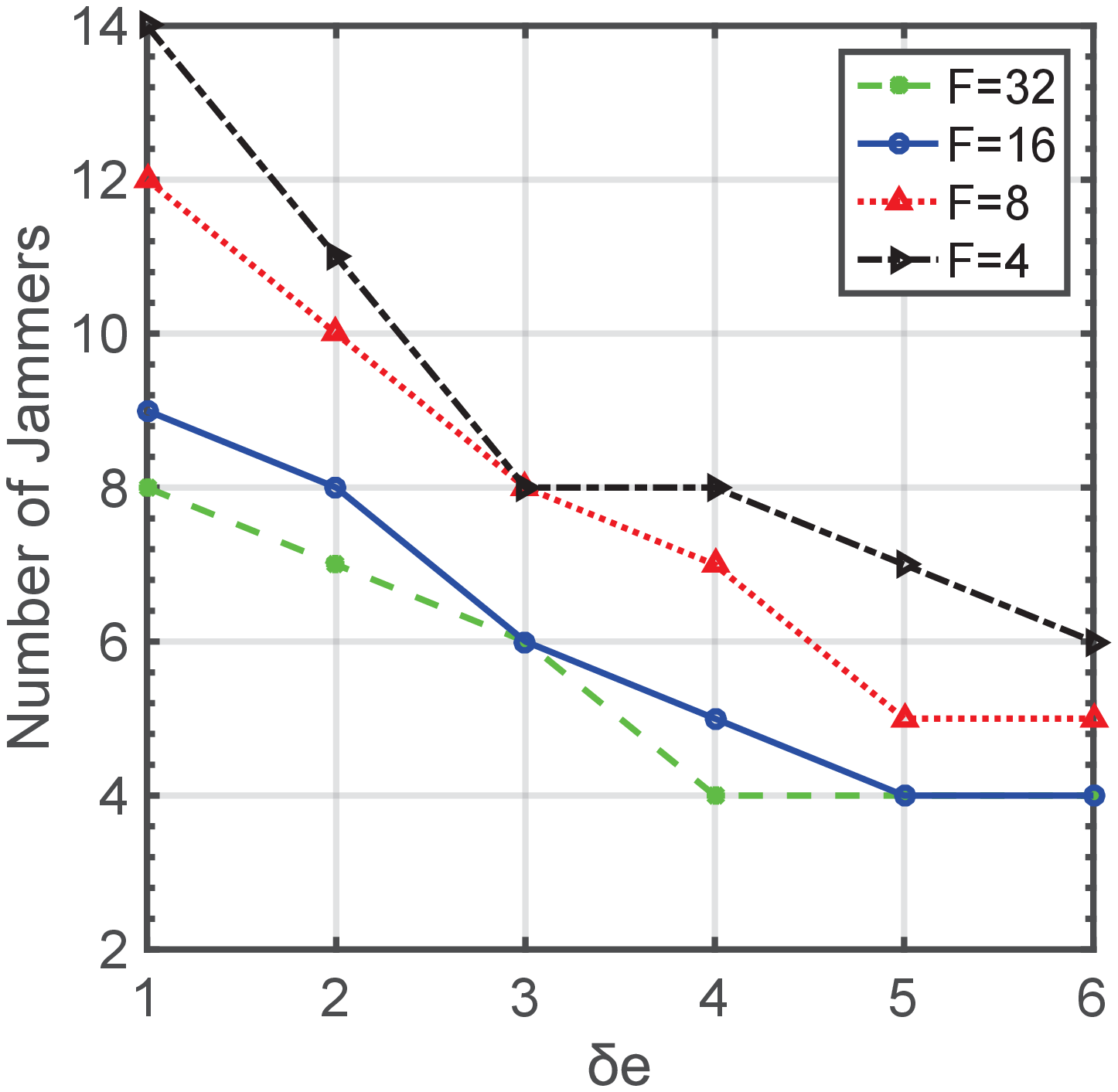}}
\end{minipage}%
\hfill
\caption{Results of simulations under proposed online algorithms.}
\end{figure*}

\section{Improved Competitive Ratios and EE under Spatial and Temporal Extensions }
\subsection{Jammer placement request  with duration}
\indent In the previous sections we assumed that requests last forever, analyzing only the spatial aspect of the problem. We now show how our results extend when each request $i$ has a duration $t_i$. After time $t_i$ an accepted request stops sending and leaves(thus, no longer causing interference).\\
\indent We show the modification for the algorithm S{\footnotesize AFE}-D{\footnotesize ISTANCE} for $r\in[0,1]$. We adapt the algorithm in the following way. It accepts a given request $i$ if and only if the safe distance $\sigma$ holds to all previously accepted requests that are active at some point time in $i$'$s$ duration.\\
\indent Our first observation is the following. If we consider a fixed point in time, an optimal solution OPT can have at most $O(\Delta^d)$ more requests than our algorithm, as this corresponds to the spatial problem. Now let $i$ be a request accepted by S{\footnotesize AFE}-D{\footnotesize ISTANCE} with smallest duration possible, that is, $t_i= 1$. Each request contained in an optimal solution that interferes with $i$ is active at least either when $i$ starts or when it stops sending. So it is sufficient to count the accepted requests in OPT at both of these points in time to upper bound the number of requests blocked by i, which is $2\cdot O(\Delta^d)$. Furthermore, a request $i$ with $t_i\leq\Gamma$ can be split into at most $\Gamma$ requests of duration 1, thus blocking at most $(\Gamma+1)\cdot O(\Delta^d)$ requests. The argumentation is similar for other polynomial power assignments and results in an additional factor of $\Gamma$ in all previously shown bounds.\\
\indent In the case of multiple channels, for $k = k'\cdot k''$, clustering of requests w.r.t. similar length and duration values can be used to improver the ratio for our algorithm S{\footnotesize AFE}-D{\footnotesize ISTANCE} to $O(k\cdot \Gamma^{1/k'}\Delta^{(d/2k'')+\epsilon})$. Choosing $k = log\Gamma\cdot log\Delta$, R{\footnotesize ANDOM}S{\footnotesize AFE}-D{\footnotesize ISTANCE} becomes $O(log\Gamma\cdot log\Delta0)$-competitive.

\subsection{Multiple Channels}
\indent In this section we show how to generalize the algorithms above to $k$ channels and decrease their competitive ratio. We propose a $k-channel adjustment$, in which we separate the problem by using certain channels only for specific request lengths. All requests with length in $[\Delta^{(i-1)/k},\Delta^{i/k}]$ are assigned to channel $i$, for $i$ = 1,...,k, where we assign requests of length $\Delta^{i/k}$ arbitrarily to channel $i$ or $i+1$. For each channel $i$ we apply an algorithm outlined above, which makes decisions about acceptance and power of requests assigned to channel $i$. Using this separation, we effectively reduce the aspect ratio to $\Delta^{1/k}$ on each channel. If the optimum solution has to adhere to the same length separation on the channels, this would yield a denominator $k$ in the exponents of $\Delta$ of the competitive ratios. Obviously, the optimum solution is not tied to our separation, but the possible improvement due to this degree of freedom can easily be bounded by a factor $k$. This yields the following corollary.

\begin{corollary}
M{\footnotesize ULTI}-C{\footnotesize LASS} S{\footnotesize AFE}-D{\footnotesize ISTANCE} with $k$-channel adjustment is $O(k\Delta^{(d/2k)+\epsilon})$-competitive using the square-root power assignment. S{\footnotesize AFE}-D{\footnotesize ISTANCE} with $k$-channel adjustment is $O(k\Delta^{d/k})$-competitive for any polynomial assignment with $r\in[0,1]$, and $O(k\Delta^{max\{r,1-r\}\cdot d/k})$-competitive for $r\notin[0,1]$.
\end{corollary}

\section{Simulations}

\indent We conducted preliminary experiments to compare the different numbers of the eavesdroppers. The setting we have chosen is the storage/fence shown in Figure 5. The fence is of dimensions $500\times 300$ units and we placed a grid of $1\times 1$ cells in the entire region. We simulated both J{\footnotesize AM}-S{\footnotesize AFE}D{\footnotesize IST}-P{\footnotesize OWER} and J{\footnotesize AM}-L{\footnotesize IFE}M{\footnotesize AX} in this setting. For the power assignment from J{\footnotesize AM}-S{\footnotesize AFE}D{\footnotesize IST}-P{\footnotesize OWER}, we investigated the difference in number of jammers. Finally, we observed the variation in total power assigned with $\epsilon $ and $\delta$ and the number of jammers palced with $\epsilon$, $\delta$ and $\hat{P}$. We set the round number T is 300. Give the energy conversion rate $\alpha(v)= 0.4$ and the average distance $\bar{l}=4$ as the distance $l\in[1,10]$. Therefore we can conclude the energy attenuation factor $\bar{l}^{\gamma}=16$ when the $\gamma = 4$. We chosen the following values: (i)$\epsilon = \{0.1, 0.2,0.3,0.4,0.5\}$,(ii)$\delta =  \{0.5, 0.6,...,1.0\}$, (iii)$\hat{P} = \{(1/\epsilon), (2/\epsilon),..., (5/\epsilon)\}$. In both numbers of the eavesdroppers, we removed all grid points which were in the forbidden region. \\
\indent For J{\footnotesize AM}-S{\footnotesize AFE}D{\footnotesize IST}-P{\footnotesize OWER}, the decline is more steep than the J{\footnotesize AM}-L{\footnotesize IFE}M{\footnotesize AX} because of the 300 rounds cause the J{\footnotesize AM}-L{\footnotesize IFE}M{\footnotesize AX} data is obtained by repeated average. The J{\footnotesize AM}-S{\footnotesize AFE}-DistPower and J{\footnotesize AM}-L{\footnotesize IFE}-Max give the same information about the desired jammers become more and more as the eavesdroppers goes up. The last two figures about the J{\footnotesize AM}-S{\footnotesize AFE}D{\footnotesize IST}-MultiChannel talk us the multichannel give the much less desired jammers but the benefits of multichannel will reduce as the channel grows in number in two different ways based on the $N_e = 6$.

\section{Conclusion}
In this paper, we propose the first the first distributed
protocol that provides secure communication in any geographically restricted communication networks using
energy-constrained friendly jammers wirelessly powered by legitimate transmitters as energy sources. Online learning
algorithms are proposed to maximize the lifetime of jammer and met the goal of EE with heterogenous
energy demands.  Our protocol supports dynamic behaviors,
e.g., mobility, eavesdropping (communicating) completion or
addition/removal of nodes, as along the secure communication
are restricted to the storage. However, our proposed protocol
is adaptive to the situations such information is available, e.g.,
exact positions and frequency of both legitimate communications and eavesdropping behaviors, and foreseen further EE
improvements and reduced number of jammers. We provided competitive ratios for approximate algorithms in several distributed
settings, and found the multi-channel diversity is a good approach to improve the security of wireless communications.

\end{document}